\documentclass[11pt]{article}

\usepackage[margin=1.05in]{geometry}
\usepackage{amsmath,amssymb,amsthm,mathtools}
\usepackage{graphicx}
\usepackage{float}
\usepackage[usenames,dvipsnames]{xcolor}
\usepackage[colorlinks=true,linkcolor=NavyBlue,citecolor=PineGreen,urlcolor=NavyBlue]{hyperref}
\usepackage{bm}
\usepackage{enumitem}
\usepackage{comment}
\usepackage{booktabs}
\usepackage{tikz}
\usetikzlibrary{arrows.meta,positioning,decorations.pathmorphing,calc}

\newcommand{\cA}{{\cal A}}
\newcommand{\cT}{{\cal T}}
\newcommand{\bS}{{\bm S}}
\newcommand{\bn}{{\bm n}}
\newcommand{\bnu}{{\bm \nu}}
\newcommand{\bB}{{\bm B}}
\newcommand{\bb}{{\bm b}}
\newcommand{\bs}{{\bm s}}

\newcommand{\dd}{{\rm d}}
\newcommand{\ii}{{\rm i}}

\newtheorem{remark}{Remark}

\newtheorem{lemma}{Lemma}
\newtheorem{proposition}{Proposition}

\title{Driven Quantum Stars\\ as Controlled Primitives for Real-Time Spin Dynamics}
\author{Michael (Misha) Chertkov\thanks{Program in Applied Mathematics \& Department of Mathematics, University of Arizona, Tucson, AZ 85721, USA; \texttt{chertkov@arizona.edu}}}
\date{\today}

\begin{document}
\maketitle

\begin{abstract}
Quantum advantage in real-time spin dynamics should be assessed against the strongest relevant classical substitutes, not merely against the fact that the microscopic system is made of qubits.  We propose a physics-based diagnostic for this boundary: reduce a qubit spin model to a spin-Landau--Lifshitz- (LL-) classical sector and organize the residual quantum sector as controlled corrections.  The control parameter is a graph/coordination structure, here realized by a spin star with $d$ leaves and $O(1/d)$ hub-leaf couplings.  In its field- and interaction- homogeneous form the quantum star is a benchmark for the transition from LL-substitutable dynamics to genuinely quantum, discrete-sector interference.  In its fully driven form, with time-dependent local fields and time-dependent bilinear pair tensors, it is the algorithmic primitive for more general quantum-spin structures over trees and loopy graphs. For the coherent-state return amplitude, we prove an exact leaf-elimination lemma and derive a $1/d$ hierarchy for the quantum star.  L0 is a driven one-spin weak-mean-field problem; G1 is a Gaussian nonlocal-in-time influence correction pairing leaf connected two-time kernels with the hub weak two-point function.  On bounded finite-time windows, and away from zeros of the weak amplitudes (dynamical-quantum-phase-transition times), the cumulant series is ordered so that $\log\cA-\log\cA_{\rm L0}=O(1/d)$ and $\log\cA-\log\cA_{\rm L0}-\Delta_{\rm G1}=O(1/d^2)$; the exact remainders follow these orders numerically, with nested-ensemble fits over fully driven anisotropic instances giving slopes $-1.05$ and $-2.03$.  Static homogeneous, aligned, inhomogeneous, and fully driven stars provide validation rungs, and the hierarchy is benchmarked head-to-head against a temporal matrix-product influence-matrix baseline, delineating complementary regimes of coordination and coupling strength.  Unlike rank-controlled compression on a Trotter grid, the hierarchy is ordered by a physical model parameter, is formulated in continuous time, and each of its truncation levels is itself a physical theory, with the Landau--Lifshitz sector as its high-coordination limit.  The same primitive suggests Gaussian influence-message updates on larger trees and loop-corrected extensions on loopy graphs.
\end{abstract}

\tableofcontents

\section{Motivation and problem setting}
\label{sec:motivation}

Quantum computing is often motivated in two complementary ways.  The first is physical: quantum hardware should simulate quantum dynamics more naturally than classical hardware, following Feynman's original intuition and the later circuit model of Hamiltonian simulation \cite{feynman_simulating_1982,lloyd_universal_1996,nielsen_quantum_2010}.  The second is algorithmic: quantum circuits may accelerate formally encoded classical tasks, as in factoring, search, optimization, linear algebra, and related algorithmic primitives \cite{grover_fast_1996,shor_polynomial-time_1997,nielsen_quantum_2010,montanaro_quantum_2016}.  In both settings, a fair assessment of quantum advantage is incomplete unless the proposed quantum procedure is compared with the best classical substitutes available for the observables, tasks, and time windows of interest.  This principle is already central in tensor-network analyses of quantum advantage and in the broader practice of benchmarking quantum devices against rapidly improving classical algorithms \cite{kshetrimayum_quantum_2026}.  It is also familiar from phase-space and truncated-Wigner approaches to spin dynamics \cite{schachenmayer_many-body_2015,mink_hybrid_2022}, from classical Landau--Lifshitz emulation of spin-$1/2$ correlations \cite{kim_emulation_2025}, and from rigorous large-spin or mean-field semiclassical limits \cite{frohlich_semi-classical_2007}.  In particular, if a classical problem is first mapped formally into a qubit dynamics problem, it is natural to ask whether a classical approximation to that qubit dynamics already solves the original problem to the required accuracy.

Extending this perspective we ask whether a qubit spin dynamics problem admits a \emph{\bf spin-LL-classical} surrogate: a Landau--Lifshitz-type dynamics for Bloch vectors or weak spin trajectories obtained from the qubit model itself.  The formal route may use Hubbard--Stratonovich fields \cite{stratonovich_method_1957,hubbard_calculation_1959}, spin-coherent states \cite{radcliffe_properties_1971,klauder_path_1979,auerbach_interacting_1994}, or weak-value/two-boundary spin trajectories \cite{aharonov_how_1988,aharonov_two-time_2005,aharonov_two-time_2017}.\footnote{Here a spin-coherent state means a minimum-uncertainty spin state polarized along a point on the Bloch sphere.  A weak trajectory or weak value means a matrix element conditioned on both initial and final boundary states, normalized by the corresponding transition amplitude; it is generally complex and should not be confused with an ordinary expectation value.}  Since the microscopic spins remain $S=1/2$, the small parameter is not $1/S$.  It must come from structure: high coordination, weak coupling, short memory, high temperature, dephasing, or another model-specific expansion parameter.  In the star problem below the control parameter is $1/d$.

This suggests a diagnostic ladder for classical substitutability.  If the LL sector reproduces the qualitative behavior of the observable, the dynamics may be quantum microscopically but classically substitutable for that observable and time window.  If the first quantum correction is needed but remains perturbative, the dynamics is weakly quantum yet still classically correctable.  If both the LL sector and the first corrections fail qualitatively, the task enters a genuinely quantum regime for the chosen diagnostic.  In this paper the L0 approximation plays the role of the LL-classical sector, while G1 is the first Gaussian quantum-sector correction.

Fig.~\ref{fig:uniform-ll-to-quantum} illustrates this philosophy in the homogeneous star.  The model is simple enough to admit a polynomial Schur-block exact oracle, but a two-population leaf product state populates many total-leaf-spin sectors.  At early times the corrected LL description follows the exact return amplitude closely.  At longer times, on the scale $T=O(d)$, discrete-sector interference and recurrence physics become visible and the classical surrogate ceases to be qualitatively reliable.  This uniform-star experiment is therefore best understood as a controlled \emph{quantum-advantage diagnostic} : it asks when a spin-LL surrogate is enough and when genuinely quantum sector structure matters.

\begin{figure}[H]
\centering
\includegraphics[width=0.82\textwidth]{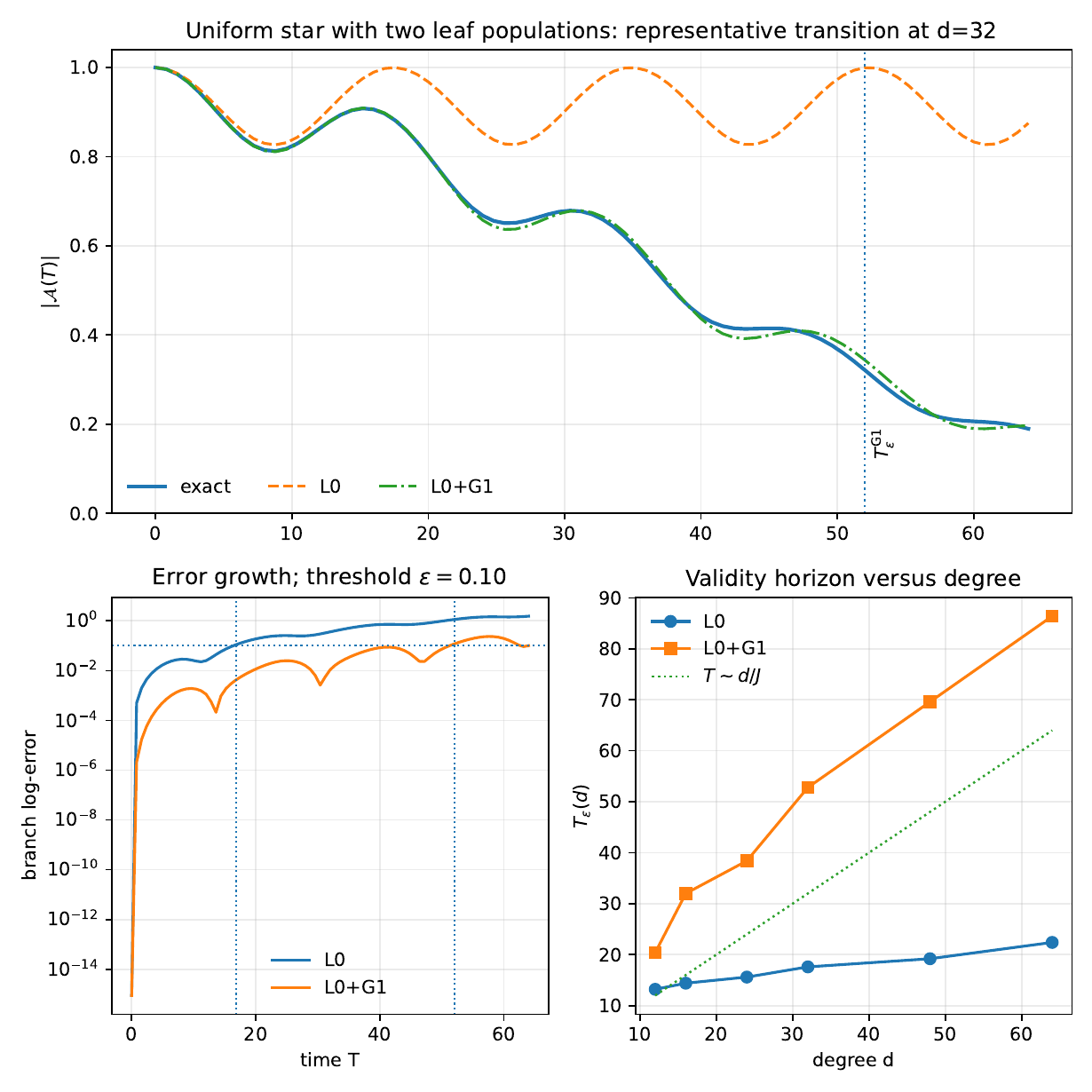}
\caption{Uniform-star diagnostic for the transition from LL-substitutable to genuinely quantum behavior.  The homogeneous star has uniform couplings $H=(J_0/d)\bS_0\cdot\sum_c\bS_c$ -- the static homogeneous special case, Eq.~\eqref{eq:homogeneous-H} of Sec.~\ref{sec:homogeneous-star}, of the general driven star Hamiltonian defined in Eq.~\eqref{eq:general-H} -- and a two-population leaf coherent product state (a deterministic $70\%/30\%$ split of the leaf Bloch directions between two fixed orientations).  The upper panel compares the exact Schur-block return amplitude with L0 and L0+G1 for $d=32$.  The lower-left panel shows the branch-continuous log-error and the threshold defining the validity horizon.  The lower-right panel shows that G1 extends the validity horizon substantially: over the scanned decade $d\in[12,64]$ the fitted growth exponents are $T_\varepsilon^{\rm L0}\sim d^{0.30}$ and $T_\varepsilon^{\rm G1}\sim d^{0.82}$, the latter approaching the linear $T=O(d/J)$ revival scale (dotted guide).  The figure is a diagnostic benchmark, not the most general algorithmic setting.  All ingredients are specified later in the paper: the return amplitude in Eq.~\eqref{eq:return-amplitude}, the model and its polynomial Schur-block exact oracle in Secs.~\ref{sec:homogeneous-star}--\ref{sec:homogeneous-coherent} and Appendix~\ref{app:oracles}, the L0 and G1 approximants in Sec.~\ref{sec:hierarchy}, and the branch-continuous log-error together with the horizon protocol in Appendix~\ref{app:numerics}.
\label{fig:uniform-ll-to-quantum}}
\end{figure}

The second role of the star is algorithmic.  The fully driven star, not the uniform benchmark, is the primitive needed for message-passing extensions.  A subtree attached to a parent spin is summarized by a scalar influence, a weak mean trajectory\footnote{By a ``weak mean trajectory'' we mean, following the weak-value and two-state-vector terminology of \cite{aharonov_how_1988,aharonov_two-time_2005,aharonov_two-time_2017}, a one-time weak value conditioned on both the initial and final amplitude boundaries, e.g.
\[
S^\alpha_w(t)=
\frac{\langle \psi_f|U(T,t)S^\alpha U(t,0)|\psi_i\rangle}
{\langle \psi_f|U(T,0)|\psi_i\rangle}.
\]
The corresponding two-time weak kernel is the connected time-ordered fluctuation around this weak mean trajectory.}, and a connected two-time kernel.  This is the quantum-spin analogue of replacing an environment by a message.  On trees this leads to a dynamic-programming structure; on loopy graphs it points toward BP/GBP and loop-corrected methodology \cite{chertkov_loop_2006,chertkov_loop_2006-1}, originally developed for probabilistic (positive-valued) densities but algebraically more general and thus recently adapted to complex-valued tensor-network constructions \cite{tindall_contracting_2026,evenbly_loop_2026,midha_beyond_2025,midha_belief_2026}.

We therefore study a spin-star graph consisting of one hub spin and $d$ leaf spins, all spin-$1/2$.  The hub-leaf couplings are scaled as $1/d$, so that the leaf-induced mean field is $O(1)$ while its connected fluctuations are suppressed by coordination.  The central observable is the coherent-state return amplitude 
\begin{equation}
\cA(T)=\langle \Psi_0|U(T,0)|\Psi_0\rangle,
\qquad
|\Psi_0\rangle=|\bn_0\rangle\otimes\bigotimes_{c=1}^d|\bnu_c\rangle,
\label{eq:return-amplitude}
\end{equation}
where $|\bn\rangle$ is a spin-coherent state polarized along the unit vector $\bn$, i.e., the normalized spin-$1/2$ state obeying $(\bn\cdot\bS)|\bn\rangle=\tfrac12|\bn\rangle$.  The unit Bloch vectors $\bn_0$ and $\bnu_c$ specify the boundary states of the hub and of the $c$-th leaf, respectively, and the same product state $|\Psi_0\rangle$ serves as both the initial and the final boundary of the amplitude.  Here and throughout,
\begin{equation}\label{eq:propagator}
U(t_2,t_1)=\cT\exp\left[-\ii\int_{t_1}^{t_2}H(t)\,\dd t\right],
\end{equation}
is the unitary propagator generated by the driven star Hamiltonian $H(t)$, written in its general form in Eq.~\eqref{eq:general-H} of Sec.~\ref{sec:general-star}, and $\cT$ denotes time ordering.  This choice is deliberate.  The return amplitude lives on a single time contour, has a direct interferometric interpretation, and is sensitive to phase-coherent many-spin dynamics.

A comment on the choice of the time contour is in order, since it fixes the class of objects used throughout the paper.  The amplitude \eqref{eq:return-amplitude} is defined on a \emph{single} forward contour $[0,T]$: it is the Loschmidt amplitude of the product state, accessible by ancilla interferometry (Hadamard-test- or Ramsey-type protocols).  Many quantities of primary physical interest -- expectation values $\langle S_0^\alpha(t)\rangle$, two-time correlation functions, dynamics in the presence of dephasing -- instead live on the doubled Keldysh closed-time-path contour, with a forward and a backward branch.  Working on the single contour is a deliberate technical simplification, not a limitation of principle: it is the minimal arena in which every structural element of the method already appears -- exact spatial elimination of the leaves, complex two-boundary (weak) trajectories and kernels, and the $1/d$ ordering of the cumulant series -- with the lightest possible bookkeeping, one world line and one set of weak objects per leaf.  Each element admits a closed-contour counterpart.  The leaves commute with one another and couple to the rest of the graph only through the hub on \emph{both} branches, so the elimination lemma of Sec.~\ref{sec:leaf-elimination} extends verbatim to the Keldysh contour, with the cumulants becoming contour ordered: the weak mean trajectory acquires a branch index and the two-time kernel becomes a $2\times2$ matrix of branch blocks -- a Feynman--Vernon influence functional -- while the $1/d$ counting is unchanged on both branches.  The doubled bookkeeping, required to promote the diagnostic from amplitudes to physical expectation values and to open-system dynamics, is deferred to a separate treatment; see Remark~\ref{rem:scope} and the outlook in Sec.~\ref{sec:outlook}.

A second orienting comment concerns the role of the star for larger graphs.  The closed number \eqref{eq:return-amplitude}, with the hub contracted against $\langle\bn_0|\cdots|\bn_0\rangle$, is the star-level version of a more general \emph{open} primitive.  In the tree recursion of Sec.~\ref{sec:trees-loops}, the identical leaf-elimination step is applied at a vertex without closing that vertex into a number: after its children are eliminated, the vertex reports not the scalar $\cA$ but its own dressed weak objects -- a scalar log-contribution, a weak mean trajectory, and a connected two-time kernel, evaluated under the children-dressed evolution -- which are passed as a message to its parent.  The closure against the boundary coherent states is performed once, at the tree root.  The star studied in this paper is thus simultaneously the simplest closed benchmark and the local open update rule.

\begin{figure}[H]
\centering
\includegraphics[width=0.76\textwidth]{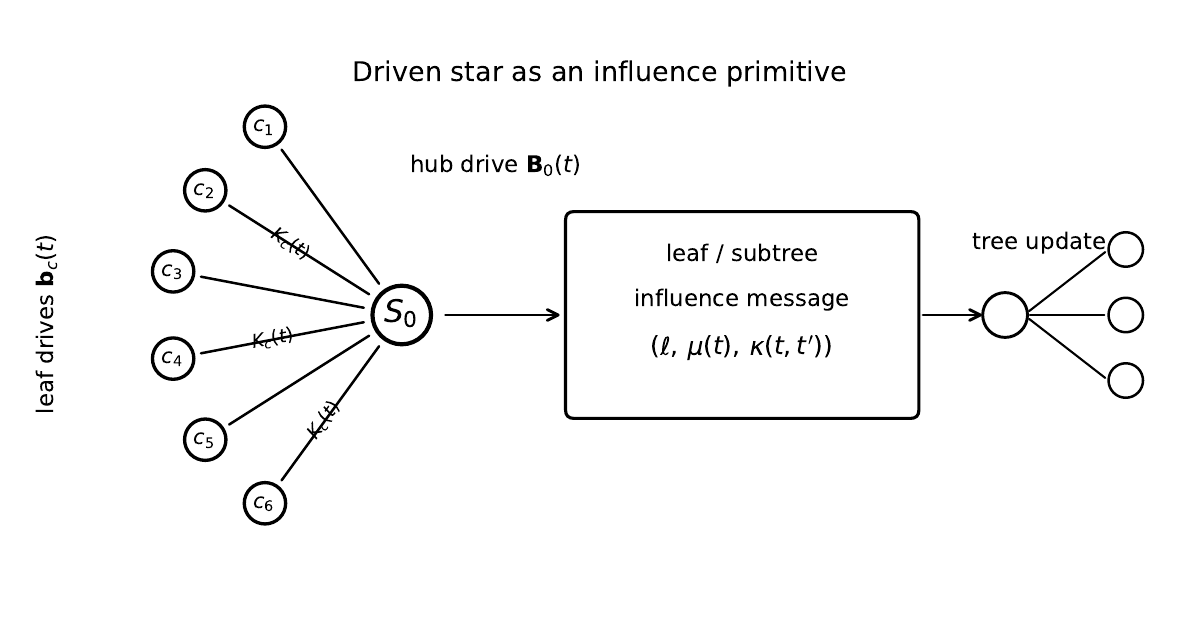}
\caption{The driven star is the algorithmic primitive.  Leaf elimination converts $d$ quantum leaves into a weak mean trajectory and a two-time influence kernel acting on the hub.  On larger trees the same object becomes a directed message from a subtree to its parent.}
\label{fig:star-primitive}
\end{figure}

Fig.~\ref{fig:star-primitive} also fixes the notation used below for the algorithmic interpretation: after elimination, a leaf or subtree contributes a scalar factor, a weak mean trajectory, and a two-time connected kernel.  Throughout the paper, \emph{elimination} means the exact contraction of the leaf world lines out of the amplitude \eqref{eq:return-amplitude}: a quantum sum over the histories of the eliminated spins, conditioned on their coherent boundary data.  It is the single-contour, fixed-boundary analogue of marginalization, performed at the amplitude level rather than at the density-matrix level -- no partial trace is involved.  All $d$ leaf spins are eliminated in one step while the hub spin is retained; leaf $c$ is summarized \emph{exactly} by its free return amplitude $\tau_c$, its weak mean trajectory $\mu_c(t)$, its connected two-time kernel $\kappa_c(t,t')$, and the tower of higher connected cumulants -- these objects are defined in Eqs.~\eqref{eq:tau}--\eqref{eq:kappa-general} of Sec.~\ref{sec:leaf-elimination}, the elimination itself is Lemma~\ref{lem:elimination}, and the $1/d$ organization and Gaussian truncation of the cumulant tower is the subject of Sec.~\ref{sec:hierarchy}.  
The open tree-level use of the same step, and the single-contour nature of all objects, were discussed above.  Fig.~\ref{fig:method-summary} assembles the full pipeline of the method -- the continuous-time model, the exact elimination, the $1/d$-ordered truncation with its Landau--Lifshitz asymptotic, and the closed/open closure -- and anticipates the structural comparison with temporal tensor networks in Sec.~\ref{sec:positioning} (Table~\ref{tab:message-comparison}) and Sec.~\ref{sec:tmps-baseline}.

\begin{figure}[H]
\centering
\resizebox{0.98\textwidth}{!}{%
\begin{tikzpicture}[
  stage/.style={draw=black!60, rounded corners=3pt, align=center, inner sep=6pt, font=\small, fill=black!3},
  lab/.style={font=\footnotesize\itshape, text=black!70, align=center},
  note/.style={font=\footnotesize, text=PineGreen!80!black, align=center},
  flow/.style={-{Stealth[length=2.6mm]}, thick, black!70},
  ]
\begin{scope}[shift={(0,0)}]
  \node[circle,draw,fill=NavyBlue!25,inner sep=2.4pt] (hub) at (0,0) {};
  \foreach \a in {90,141,193,244,296,347}{
    \node[circle,draw,fill=black!12,inner sep=1.6pt] (l\a) at (\a:0.95) {};
    \draw[black!55] (hub)--(l\a);}
  \node[font=\scriptsize] at (0,-0.34) {$\bS_0$};
  \node[font=\scriptsize] at (75:1.28) {$\bS_c$};
  \node[font=\scriptsize, black!70] at (18:1.52) {$K_c(t)$};
\end{scope}
\node[stage, anchor=north, text width=3.1cm] (s1) at (0,-1.55)
  {\textbf{I. driven star}\\[1pt] $H(t)$: fields $\bB_0(t)$, $\bb_c(t)$,\\ pair tensors $K_c(t)=O(1/d)$\\[1pt] {\footnotesize\itshape continuous time}};
\node[stage, right=1.25cm of s1, text width=3.5cm] (s2)
  {\textbf{II. exact leaf elimination}\\ (Lemma~\ref{lem:elimination}, single contour)\\[2pt]
   per leaf: $\tau_c,\ \mu_c(t),\ \kappa_c(t,t'),\dots$\\[1pt]
   {\footnotesize\itshape continuum weak objects;}\\ {\footnotesize\itshape grid $=$ quadrature choice only}};
\draw[flow] (s1) -- node[lab, above=1pt]{exact} (s2);
\begin{scope}[shift={($(s2.north)+(0,0.85)$)}]
  \draw[black!60,-{Stealth[length=1.6mm]}] (-1.0,0)--(1.05,0) node[right,font=\scriptsize]{$t$};
  \draw[NavyBlue, thick, decorate, decoration={snake, amplitude=.5mm, segment length=2.6mm}] (-0.95,0.22)--(0.95,0.22);
  \node[font=\scriptsize, NavyBlue] at (0,0.55) {$\mu_c(t)$,\ \ $\kappa_c(t,t')$};
\end{scope}
\node[stage, right=1.25cm of s2, text width=4.05cm] (s3)
  {\textbf{III. $1/d$-ordered truncation}\\[3pt]
   {\footnotesize L0: weak LL sector, error $O(1/d)$}\\
   {\footnotesize $+\,$G1: Gaussian influence, $O(1/d^{2})$}\\
   {\footnotesize\color{black!50}$+\,$G2, \dots: higher cumulants}\\[3pt]
   {\footnotesize\itshape each level a physical theory}};
\draw[flow] (s2) -- node[lab, above=1pt]{cumulants} (s3);
\node[note, text width=4.1cm] at ($(s3.south)+(0,-0.62)$)
  {$d\to\infty$: Landau--Lifshitz\\ classical sector (physical asymptotic)};
\node[stage, right=1.25cm of s3, text width=3.6cm] (s4)
  {\textbf{IV. closure}\\[2pt]
   \emph{closed} (star): $\cA(T)$ --- benchmark, diagnostic\\[2pt]
   \emph{open} (tree): message $\mathfrak m_{i\to p}=(\ell,\mu,\kappa)$, additive aggregation};
\draw[flow] (s3) -- (s4);
\begin{scope}[shift={($(s4.north)+(0,0.8)$)}]
  \node[circle,draw,fill=black!12,inner sep=1.4pt] (t1) at (-0.55,-0.2) {};
  \node[circle,draw,fill=black!12,inner sep=1.4pt] (t2) at (-0.15,-0.2) {};
  \node[circle,draw,fill=NavyBlue!25,inner sep=1.9pt] (ti) at (-0.35,0.25) {};
  \node[circle,draw,fill=black!25,inner sep=1.9pt] (tp) at (0.45,0.62) {};
  \draw[black!55] (t1)--(ti)--(t2); \draw[-{Stealth[length=1.8mm]},PineGreen!80!black,thick] (ti)--(tp);
  \node[font=\scriptsize,PineGreen!80!black] at (0.35,0.1) {$\mathfrak m_{i\to p}$};
\end{scope}
\end{tikzpicture}}%
\caption{Summary of the method.  (I) The model is the fully driven star in \emph{continuous} time.  (II) Leaf elimination (Lemma~\ref{lem:elimination}) is exact and produces continuum weak objects per leaf -- a scalar, a weak mean trajectory, and a connected two-time kernel; the time grid enters only as a quadrature choice at numerical evaluation.  (III) The cumulant tower is ordered by the physical model parameter $1/d$, and each truncation level is itself a physical theory: L0 is the weak Landau--Lifshitz classical sector, G1 the Gaussian (noise-plus-response) influence correction; the $d\to\infty$ limit is the LL dynamics.  (IV) Closing the hub against its boundary state gives the benchmark amplitude $\cA(T)$; leaving it open gives the additive message passed to a parent in the tree recursion of Sec.~\ref{sec:trees-loops}.  Table~\ref{tab:message-comparison} contrasts this pipeline with rank-controlled temporal-MPS compression.}
\label{fig:method-summary}
\end{figure}
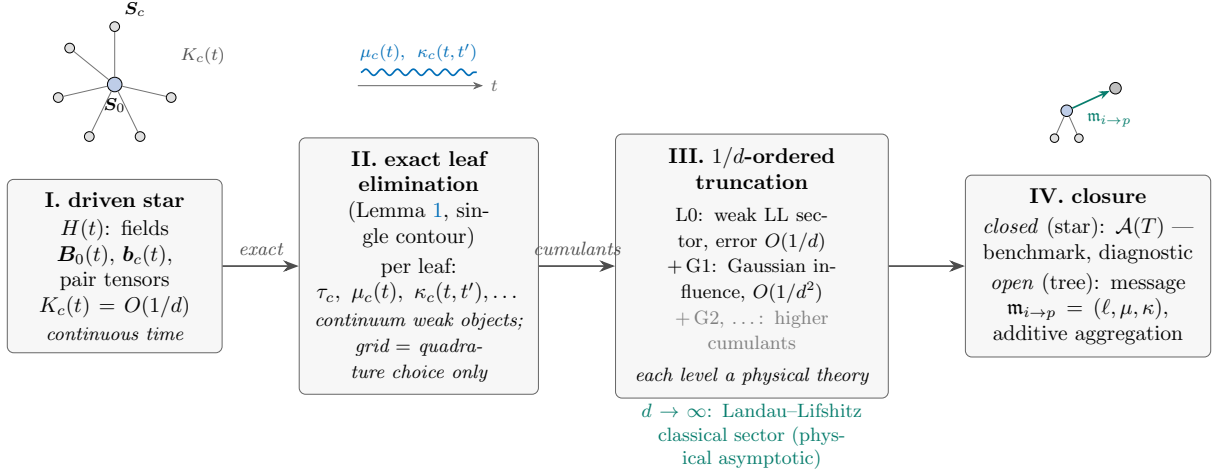

\section{Brief positioning relative to prior work}
\label{sec:positioning}

The construction above touches several literatures, but the intended contribution is focused.  It is not that a star graph can be solved in special cases, nor that dynamic programming is possible on trees.  The contribution is to combine three ideas in one controlled spin-$1/2$ setting: an emergent spin-LL-classical sector, a $1/d$ expansion of quantum corrections around it, and a fully driven star primitive suitable for message passing.

\paragraph{From quantum Heisenberg dynamics to LL-classical surrogates.}
The functional-integral route to high-temperature Heisenberg dynamics already appeared in the 1994--95 work of Chertkov and Kolokolov \cite{chertkov_long-time_1994,chertkov_equilibrium_1995}, building on Kolokolov's single-spin transformations \cite{kolokolov_functional_1986,kolokolov_functional_1990}.  There the infinite-temperature spin correlator was related, under controlled physical assumptions, to an average over classically evaluated vector fields.  The modern spin-transport literature arrived at a closely related conclusion by a different route: numerical and experimental studies of spin-$1/2$ chains revealed anomalous high-temperature transport \cite{znidaric_transport_2011,ljubotina_spin_2017,ljubotina_kardarparisizhang_2019,scheie_detection_2021,wei_quantum_2022}, and generalized-hydrodynamic and classical-field interpretations connected the quantum dynamics to classical Landau--Lifshitz solitons and KPZ scaling \cite{gopalakrishnan_kinetic_2019,de_nardis_superdiffusion_2020,das_nonlinear_2020}.  We use this history as motivation, not as an assumption: the paper asks how far such LL-classical substitutability can be made quantitative on a finite graph and where the leading quantum corrections become necessary.

\paragraph{Coherent states, large-$S$, and why $1/S$ is not the control parameter here.}
Spin coherent states provide the natural bridge between spin operators and classical Bloch-vector trajectories \cite{radcliffe_properties_1971,klauder_path_1979,auerbach_interacting_1994,polkovnikov_phase_2010}.  Schematically, for a single spin driven by a field $\bm\Xi(t)$,
\begin{equation}
\langle \bn_f|\cT e^{-\ii\int_0^T \bm\Xi(t)\cdot\bS\,\dd t}|\bn_i\rangle
=\int_{\bn(0)=\bn_i}^{\bn(T)=\bn_f}\!\mathcal D\bn\; e^{\ii S_{\rm B}[\bn]-\ii\int_0^T H_{\rm cl}(\bn(t),t)\,\dd t},
\label{eq:coherent-state-schematic}
\end{equation}
where $S_{\rm B}$ is the Berry phase and the stationary path obeys a Landau--Lifshitz equation, $\dot{\bn}=\bn\times\bm\Xi$ in the simplest convention \cite{landau_theory_1935,lakshmanan_fascinating_2011}.  In the usual semiclassical use of \eqref{eq:coherent-state-schematic}, large spin $S$ suppresses fluctuations.  That is not our regime: all microscopic spins are qubits.  The star instead supplies a different small parameter.  With $K_c=O(1/d)$, the hub sees an $O(1)$ mean field from many leaves, while the connected $k$-th cumulant of the summed leaf influence scales as $O(d^{1-k})$.  Thus L0 is the LL-classical sector selected by large coordination, and G1 is the leading finite-coordination quantum fluctuation.

\paragraph{Phase-space spin dynamics and classical substitutes.}
Discrete and continuous Wigner approaches to spin-$1/2$ dynamics are important neighboring classical substitutes \cite{schachenmayer_many-body_2015,wurtz_cluster_2018,kunimi_performance_2021,mink_hybrid_2022}.  They also show that spin-$1/2$ dynamics can be accurately approximated for selected observables when interaction range, coordination, or coarse-graining suppresses genuinely quantum interference.  Our hierarchy is complementary: instead of postulating a phase-space sampling rule, we derive the LL mean and the first connected weak two-time correction by eliminating leaves of the qubit Hamiltonian itself.  This makes the success or failure of the surrogate observable-dependent and testable through the log-amplitude error.

\paragraph{Dynamic programming, quantum BP, and space--time tensor networks.}
The tree structure connects the work to belief propagation (BP) ideas \cite{bethe_statistical_1935,pearl_probabilistic_1988,yedidia_constructing_2003,wainwright_graphical_2008}.  Classical BP is exact on trees, while loop corrections quantify what is lost on loopy graphs \cite{chertkov_loop_2006,chertkov_loop_2006-1,chertkov_gauges_2020,chertkov_inferlo_2024}.  Quantum versions appear in several forms: imaginary-time quantum cavity methods on Bethe lattices \cite{laumann_cavity_2008,krzakala_path-integral_2008}, Hastings' thermal quantum belief propagation \cite{hastings_quantum_2007}, and quantum-message BP for communication and inference settings \cite{poulin_belief-propagation_2008,rengaswamy_belief_2021}.  Real-time simulation is often handled by Suzuki--Trotter decompositions \cite{trotter_product_1959,suzuki_generalized_1976}, which convert the problem into a space--time tensor network.  A powerful modern line eliminates an environment exactly and compresses the resulting temporal object -- the influence matrix or process tensor -- as a matrix product state (MPS) in the time direction \cite{banuls_matrix_2009,strathearn_efficient_2018,lerose_influence_2021,sonner_influence_2021}; related influence-functional and tensor-network BP methods compress the temporal direction by MPS or cluster corrections \cite{park_simulating_2025,tindall_contracting_2026,evenbly_loop_2026,midha_beyond_2025,midha_belief_2026}.  Our star primitive shares the elimination-then-compression \emph{logic} of this line, but the sharing stops at the logic: the two constructions differ in three structural respects, summarized in Fig.~\ref{fig:method-summary} and Table~\ref{tab:message-comparison}:
\begin{enumerate}
\item \underline{\emph{What controls the truncation}}.  The influence-matrix hierarchy is ordered by the bond dimension $\chi$ -- an algorithmic rank with no physical meaning of its own, whose required value is discovered a posteriori from the temporal entanglement of the computed object.  The cumulant hierarchy is ordered by $1/d$ -- a model parameter read off the graph before any computation -- and every truncation level is itself a physical theory: L0 is the weak Landau--Lifshitz classical sector, G1 is the Gaussian (noise-plus-response) influence correction, and the $d\to\infty$ limit is the LL dynamics naturally attributed to spins.  This is also why the hierarchy doubles as the quantum-advantage diagnostic of Sec.~\ref{sec:motivation}: the truncation ladder \emph{is} the classical-to-quantum ladder, a reading that a rank expansion cannot support.

\item \underline{\emph{Where the objects live}.}  The temporal tensor network exists only after Trotterization: the influence matrix is defined on a specific grid, changing or refining the grid rebuilds every tensor and every compression, and the discretization error is intrinsic to the message itself.  Our primitive is formulated in continuous time -- Lemma~\ref{lem:elimination} is a continuum identity, and $(\ell,\mu(t),\kappa(t,t'))$ are functions -- so the grid enters only as a quadrature choice at evaluation time, and adaptive, nonuniform, or spectral representations of the message require no rebuilding of the underlying objects. 

\item \underline{\emph{How messages compose}:} cumulant messages aggregate additively, at a size independent of $d$, while influence matrices compose multiplicatively, with recompression after every member.  None of this makes the influence-matrix message weaker where it is strong -- it is nonperturbative in the coupling and systematically improvable in $\chi$ at fixed degree -- and Sec.~\ref{sec:tmps-baseline} makes the comparison quantitative on the driven star, including the measured crossover between the two regimes.  
\end{enumerate}
A further structural remark: Trotterizing even a tree makes the {\em space--time} network loopy -- each hub--leaf edge generates a ladder of temporal four-cycles -- so plain BP on the discretized network is already approximate there, whereas the continuous-time leaf elimination of Sec.~\ref{sec:leaf-elimination} resums all hub--leaf temporal ladders exactly and confines the approximation to the cumulant truncation.

\begin{table}[t]
\centering
\small
\begin{tabular}{@{}p{0.185\textwidth}p{0.39\textwidth}p{0.35\textwidth}@{}}
\toprule
 & \textbf{Gaussian cumulant message \hspace{1.5cm} (this work)} & \textbf{Temporal-MPS influence matrix} \\
\midrule
truncation parameter & coordination $1/d$: a model parameter, read off the graph (describing interactions of quantum spins) a priori & bond dimension $\chi$: an algorithmic rank, set a posteriori by temporal entanglement \\[2pt]
meaning of levels & each level a physical theory: L0 $=$ weak-LL classical sector; G1 $=$ Gaussian (noise$+$response) influence; $d\to\infty$ limit $=$ LL & intermediate $\chi$ carry no independent physical identity; $\chi=2^d$ exact for the discretized dynamics \\[2pt]
error control & a priori ordering $O(1/d)$, $O(1/d^2)$ (Proposition~\ref{prop:remainder}); verified slopes $-1.0$, $-2.0$ & a posteriori, via discarded singular-value weight \\[2pt]
time & continuous: $\mu(t)$, $\kappa(t,t')$ are functions; grid $=$ quadrature; adaptive/nonuniform refinement without rebuilding & discrete: network defined by the Trotter grid; grid change rebuilds all tensors; Trotter floor intrinsic to the message \\[2pt]
composition & additive: sub-trees aggregate by summation; size independent of $d$ & multiplicative: MPO product per member, then recompression \\[2pt]
preferred regime & high coordination, weak-to-moderate coupling; trees & small or strongly coupled clusters; nonperturbative accuracy in the coupling \\
\bottomrule
\end{tabular}
\caption{Two truncations of the same exact influence.  Structural comparison of the Gaussian cumulant message developed in this paper with the temporal-MPS influence-matrix message; the quantitative head-to-head comparison, including the measured crossover $d_*(\chi)$, is in Sec.~\ref{sec:tmps-baseline} and Fig.~\ref{fig:tmps-baseline}.}
\label{tab:message-comparison}
\end{table}

\section{The general driven star}
\label{sec:general-star}

The main object of the paper is the time-dependent pairwise spin-star Hamiltonian
\begin{equation}
H(t)=\bB_0(t)\cdot\bS_0+
\sum_{c=1}^d\left(\bb_c(t)\cdot\bS_c+\bS_0^{\mathsf T}K_c(t)\bS_c\right),
\qquad \|K_c(t)\|=O(1/d).
\label{eq:general-H}
\end{equation}
Here $K_c(t)$ is a real $3\times3$ matrix. This includes isotropic exchange, anisotropic exchange, XXZ/Ising limits, Dzyaloshinskii--Moriya terms, and symmetric anisotropy. The fields $\bB_0(t)$ and $\bb_c(t)$ are arbitrary bounded functions of time. Constants in the Hamiltonian only contribute phases and will be ignored.  The spin-star geometry itself has a long history as a central-spin and spin-bath model in open-system and entanglement studies \cite{prokofev_theory_2000,breuer_non-markovian_2004,hutton_mediated_2004,gaudin_diagonalisation_1976}; the present work differs in observable (a single-contour return amplitude rather than the two-contour/Keldysh dynamics), in the fully driven anisotropic setting, and in the $1/d$ organization of the corrections.

The return amplitude \eqref{eq:return-amplitude} is generated by this Hamiltonian through the propagator \eqref{eq:propagator}. The scaling \(K_c=O(1/d)\) in Eq.~(\ref{eq:general-H}) is the central assumption.  Anticipating the leaf elimination below, each leaf contributes a weak mean trajectory \(\mu_c(t)\), and the hub sees the summed weak mean field
\begin{equation}
\mathbf B_{\rm eff}(t)=\mathbf B_0(t)+
\sum_{c=1}^d K_c(t)\,\mu_c(t).
\end{equation}
The \(1/d\) scaling keeps this field finite, while the summed connected cumulants are ordered as
\[
\sum_c K_c^2\,\kappa_c = O(1/d),\qquad
\sum_c K_c^3\,\kappa_c^{(3)} = O(1/d^2),
\]
and similarly for higher cumulants.  No assumption of time independence, homogeneity, or integrability is made in the formulation.

In the statements below, the \emph{finite-time bounded-drive regime} means that \(\bB_0(t)\), \(\bb_c(t)\), \(K_c(t)\), and a finite number of their derivatives are uniformly bounded on \([0,T]\), with \(\|K_c(t)\|\le C/d\).  Error estimates are understood on fixed finite time windows and on continuous logarithmic branches of the boundary amplitudes used to define the weak trajectories and weak kernels; in particular, the relevant transition amplitudes are assumed not to cross zero.  This assumption has physical content: the return amplitude \eqref{eq:return-amplitude} is a Loschmidt amplitude of a product state, so its zeros are exactly the Fisher zeros signalling dynamical quantum phase transitions (DQPTs) \cite{heyl_dynamical_2013,heyl_dynamical_2018}. In other words, DQPT times mark the boundaries of the domain in which perturbative treatments -- and the present $1/d$ cumulant expansion in particular -- are uniformly controlled.  As such a time is approached, transition amplitudes of the same Loschmidt family -- the free-leaf and weak-mean-field hub amplitudes appearing in the denominators of the weak objects \eqref{eq:mu}--\eqref{eq:kappa-general} -- tend to zero, the weak trajectories and kernels develop poles, hypothesis (H2) of Proposition~\ref{prop:remainder} fails, and the constants in the error estimates diverge; between consecutive DQPT times, on each continuous branch of the logarithm, the hierarchy is expected to remain controlled.  How a finite-size qubit system, and the finite-$d$ hierarchy with it, behaves as such a transition is approached -- how fast the accuracy degrades, over what time window, and whether the expansion can be repaired across the transient -- is beyond the scope of this work and is posed as a target for path forward discussion in Sec.~\ref{sec:outlook}.

\section{Exact leaf elimination}
\label{sec:leaf-elimination}

For each leaf define its free driven one-spin propagator
\begin{equation}
U_c(t_2,t_1)=\cT\exp\left[-\ii\int_{t_1}^{t_2}\bb_c(s)\cdot\bS\,\dd s\right].
\label{eq:leaf-prop}
\end{equation}
The associated leaf return amplitude, one-time weak mean trajectory, and connected two-time weak kernel are
\begin{align}
\tau_c&=\langle\bnu_c|U_c(T,0)|\bnu_c\rangle,
\label{eq:tau}\\
\mu_c^\beta(t)&=
\frac{\langle\bnu_c|U_c(T,t)S^\beta U_c(t,0)|\bnu_c\rangle}{\tau_c},
\label{eq:mu}\\
\kappa_c^{\beta\delta}(t,t')&=
\frac{
\langle\bnu_c|
U_c(T,0)\,
\cT\!\left[S_{I,c}^\beta(t)S_{I,c}^\delta(t')\right]
|\bnu_c\rangle
}{\tau_c}
-\mu_c^\beta(t)\mu_c^\delta(t'),
\label{eq:kappa-general}
\end{align}
where
\[
S_{I,c}(t)=U_c(0,t)S\,U_c(t,0)
\]
is the leaf spin in the interaction picture defined by the free leaf drive \(\bb_c(t)\).  These are one-spin, hence \(2\times2\), objects and are cheap to compute numerically even for fully time-dependent drives; Appendix~\ref{app:weak-objects} gives the explicit forward--backward propagator formulas used in the code.

The exact spatial elimination of the leaves is the content of the following lemma.  Because the source that enters each leaf functional is operator valued -- it contains the hub spin -- the exponential of $W_c$ must be understood as a formal identity: both sides are defined by their expansions in powers of the hub--leaf coupling, with every hub insertion generated by the expansion kept under a single global time-ordering symbol.  With this convention the identity is exact term by term, to all orders.

\begin{lemma}[Exact leaf elimination]
\label{lem:elimination}
Let $H(t)$ be the driven star Hamiltonian \eqref{eq:general-H} with bounded drives on $[0,T]$, and let $|\Psi_0\rangle$ be the product coherent state of Eq.~\eqref{eq:return-amplitude}.  Then $\cA(T)$ is
\begin{equation}
\cA(T)=\left(\prod_{c=1}^d\tau_c\right)
\left\langle \bn_0 \,\middle|\,
\cT\exp\left[
-\ii\int_0^T \bB_0(t)\cdot\bS_0\,\dd t
+\sum_{c=1}^d W_c\big[K_c^{\mathsf T}(\cdot)\bS_0(\cdot)\big]
\right]
\,\middle|\, \bn_0\right\rangle ,
\label{eq:exact-elimination}
\end{equation}
where $\tau_c$  is the free leaf return amplitude \eqref{eq:tau} and $W_c$ is the connected cumulant-generating functional of leaf $c$ defined by
\begin{equation}
\exp \left(W_c[h_c]\right)=
\frac{
\left\langle\bnu_c\,\middle|\,
U_c(T,0)\,
\cT\exp\left[-\ii\int_0^T h_c^\beta(t)S_{c,I}^\beta(t)\,\dd t\right]
\,\middle|\,\bnu_c\right\rangle
}{\tau_c},
\label{eq:leaf-W-def}
\end{equation}
where $S_{c,I}(t)=U_c(0,t)S_c\,U_c(t,0)$ is the leaf spin in the interaction picture generated by the free drive
\(\bb_c(t)\).
Both sides in Eq.~(\ref{eq:exact-elimination}) are understood as formal power series in the couplings $K_c$ under the global time ordering.
\end{lemma}

\begin{proof}
Pass to the interaction picture with respect to the free leaf Hamiltonian $H_{\rm leaf}(t)=\sum_c\bb_c(t)\cdot\bS_c$.  Then
\begin{align*}
\cA(T)\!& =\!\Big\langle \bn_0\Big|\otimes\bigotimes_c\Big\langle\bnu_c\Big|\,
\\ & U_{\rm leaf}(T,0)\;
\cT\exp\left[-\ii\int_0^T\Big(\bB_0(t)\cdot\bS_0
\!+\!\sum_c S_0^\alpha K_c^{\alpha\beta}(t) S^\beta_{c,I}(t)\Big)\dd t\right]
\,\Big|\bn_0\Big\rangle\otimes\bigotimes_c\Big|\bnu_c\Big\rangle,
\end{align*}
with $S_{c,I}$ the interaction-picture leaf spin of Sec.~\ref{sec:leaf-elimination} and $U_{\rm leaf}=\prod_c U_c(T,0)$.  Expand the time-ordered exponential in powers of the couplings.  Each term is a time-ordered product of hub factors $S_0^\alpha(t_r)$ and leaf factors $S_{c,I}^\beta(t_r)$.  Operators belonging to different leaves act on distinct tensor factors and commute, and each interaction term couples the hub to exactly one leaf; hence, for every fixed string of hub insertions, the leaf expectation factorizes into a product over $c$ of one-leaf matrix elements of the form appearing in Eq.~\eqref{eq:leaf-W-def}.  Resumming the one-leaf series for each $c$ reproduces $\tau_c\exp(W_c[h_c])$ with the operator-valued source $h_c^\beta(t)=K_c^{\alpha\beta}(t)S_0^\alpha(t)$, evaluated by the linked-cluster (connected-cumulant) theorem for time-ordered moments; the hub factors extracted in this way remain ordered by the single global $\cT$ symbol, which is Eq.~\eqref{eq:exact-elimination}.  Every manipulation above is an identity of formal series coefficients; no truncation has been performed.
\end{proof}

\begin{remark}[Scope of the exactness]
\label{rem:scope}
Three delimitations of Lemma~\ref{lem:elimination} should be kept explicit:
\begin{enumerate} 
\item \underline{\emph{Single contour.}}  The statement and the proof concern the single-contour amplitude \eqref{eq:return-amplitude}, with fixed coherent boundary data on every spin.  As discussed earlier in the manuscript (in the paragraph following Eq.~\eqref{eq:return-amplitude}), the same commutation-and-factorization argument extends verbatim to the Keldysh closed-time-path contour, yielding an exact Feynman--Vernon influence functional per leaf with contour-ordered cumulants and unchanged $1/d$ counting; that doubled bookkeeping is deferred (Sec.~\ref{sec:outlook}).

\item \underline{\emph{Star versus tree.}}  The elimination step is exact whenever the eliminated subsystem couples to the remainder through a single vertex, and it therefore applies recursively on any tree.  What is special about the star is that a single step completes the calculation, so exactness can be retained at the level of closed formulas.  On a deeper tree, keeping the recursion exact would require propagating the \emph{full} cumulant functional of every eliminated subtree; a finite message must be truncated -- in this paper, at the Gaussian pair $(\mu,\kappa)$ -- at every generation, so beyond the star the scheme is approximate by construction (Sec.~\ref{sec:trees-loops}).

\item \underline{\emph{Exact rewriting versus practical scheme.}}  Even on the star, the identity \eqref{eq:exact-elimination} is the exact \emph{input} to a controlled approximation, not the computational scheme itself: the algorithm retains only the first two cumulants of each $W_c$, so the computed amplitude is approximate, with the error ordered by Proposition~\ref{prop:remainder} -- $O(1/d)$ after L0 and $O(1/d^2)$ after L0+G1.  ``Exact'' refers to the spatial elimination, which confines all approximation to the transparent, $1/d$-ordered truncation of the cumulant series.
\end{enumerate}
\end{remark}

Here \(\bn_0\in S^2\) is the Bloch-sphere direction specifying the initial and final hub coherent state \(|\bn_0\rangle\) in the return amplitude, while \(\bS_0(t)\) denotes the hub spin insertion inside the remaining time-ordered hub amplitude.  The functional \(W_c[h_c]\) is the connected time-ordered cumulant-generating functional of the \(c\)-th driven leaf, normalized by its return amplitude \(\tau_c\), and evaluated at the hub-dependent source \(h_c(t)=K_c^{\mathsf T}(t)\bS_0(t)\).  Expanding \(W_c\) gives the \(1/d\) hierarchy derived in Sec.~\ref{sec:hierarchy} and Appendix~\ref{app:cumulant-derivation}.  Its first two cumulants are precisely the one-time weak mean trajectory \eqref{eq:mu} and the connected two-time weak kernel \eqref{eq:kappa-general}.

In the static zero-field leaf case, \eqref{eq:kappa-general} reduces to the familiar spin-coherent covariance and response kernel
\begin{equation}
\kappa_c^{\alpha\beta}(t,t')=\frac14(\delta^{\alpha\beta}-\nu_c^\alpha\nu_c^\beta)
+\frac{\ii}{4}\varepsilon^{\alpha\beta\gamma}\nu_c^\gamma\,\mathrm{sgn}(t-t').
\label{eq:static-kappa}
\end{equation}
The two parts of this kernel play distinct physical roles: the real symmetric part is the coherent-state covariance -- the noise carried by the boundary state -- while the imaginary antisymmetric part, proportional to $\mathrm{sgn}(t-t')$, is the response carried by the spin commutation relations.  In the Keldysh framework these two roles are usually distributed over the closed-time-path branch structure (the Keldysh and retarded/advanced components of the contour-ordered kernel).  Here both are carried by a single kernel on the single contour, because the weak objects are conditioned on the final as well as the initial boundary state; this is the single-contour counterpart of the noise/response split, in the sense of the contour discussion following Eq.~\eqref{eq:return-amplitude} and of Remark~\ref{rem:scope}(i), where the same kernel becomes the $2\times2$ branch-block matrix of the Feynman--Vernon functional.

\section{The $1/d$ hierarchy: L0 and G1}
\label{sec:hierarchy}

We now derive the formulas used in the simulations.  The derivation has three steps.  First, the leaves are eliminated exactly as driven one-spin systems.  Second, the resulting exact leaf influence functionals are expanded in connected time-ordered cumulants.  Third, the cumulant series is truncated according to the $1/d$ counting.  Only the last step is an approximation; the preceding rewriting is exact.

Let
\begin{equation}
H_{\rm leaf}(t)=\sum_{c=1}^d \bb_c(t)\cdot\bS_c,
\qquad
V(t)=\sum_{c=1}^d S_0^\alpha K_c^{\alpha\beta}(t)S_c^\beta,
\end{equation}
so that the full star Hamiltonian \eqref{eq:general-H} splits as $H(t)=\bB_0(t)\cdot\bS_0+H_{\rm leaf}(t)+V(t)$: the first term is the driven hub, $H_{\rm leaf}$ collects the free leaf drives, and $V$ is the hub--leaf coupling that the elimination removes.

After factoring out the free leaf evolution, each leaf sees the hub through a time-dependent source.  For the \(c\)-th leaf this source is
\begin{equation}
h_c^\beta(t)=K_c^{\alpha\beta}(t)S_0^\alpha(t),
\label{eq:operator-source}
\end{equation}
with summation over repeated spin indices.  In the full star amplitude this is an
operator-valued source, because \(S_0^\alpha(t)\) acts on the hub Hilbert space;
the remaining hub operators are kept under the global time-ordering symbol.

Let \(\bnu_c\in S^2\) denote the Bloch-sphere direction of the initial and final
coherent state \(|\bnu_c\rangle\) of leaf \(c\).  The corresponding free leaf
return amplitude is the scalar defined in Eq.~(\ref{eq:tau}). For a test source \(h_c(t)\), the normalized leaf influence functional defined in Eq.~(\ref{eq:leaf-W-def}).
The functional \(W_c[h_c]\) is therefore normalized so that
\(W_c[0]=0\); its first functional derivative gives the one-time weak mean
trajectory \(\mu_c(t)\), and its second connected derivative gives the two-time
weak kernel \(\kappa_c(t,t')\).

Equivalently, one may insert spin coherent-state resolutions of the identity and
write the same one-spin object as a coherent-state path integral.  After the leaf
influence functionals are assembled into the remaining hub amplitude, the L0
sector gives the spin-Landau--Lifshitz saddle driven by the summed weak mean
field, while the higher connected cumulants describe the quantum fluctuation
sector around it.  We use the operator form because it is shorter and keeps the
spin-\(1/2\) one-spin objects exact.

The cumulant expansion of \eqref{eq:leaf-W-def} is
\begin{align}
W_c[h_c]
&= -\ii\int_0^T h_c^\beta(t)\mu_c^\beta(t)\,\dd t
-\frac12\int_0^T\!\int_0^T
h_c^\beta(t)h_c^\delta(t')\kappa_c^{\beta\delta}(t,t')\,\dd t\,\dd t'
+W_c^{(\ge 3)}[h_c],
\label{eq:W-cumulant}\\
W_c^{(\ge3)}[h_c]
&=\sum_{k\ge3}\frac{(-\ii)^k}{k!}
\int C_c^{\beta_1\cdots\beta_k}(t_1,\ldots,t_k)
\prod_{r=1}^k h_c^{\beta_r}(t_r)\,\dd t_r.
\label{eq:W-higher}
\end{align}
Here $\mu_c$ and $\kappa_c$ are exactly the weak one- and connected two-time objects in Eqs.~\eqref{eq:mu}--\eqref{eq:kappa-general}; $C_c^{\beta_1\cdots\beta_k}$ denotes the connected $k$-point weak cumulant of the driven leaf.  Substituting the first term of \eqref{eq:W-cumulant} into the hub time-ordered exponential simply shifts the hub field.  This gives the L0 effective field
\begin{equation}
B_{\rm eff}^{\alpha}(t)=B_0^\alpha(t)+\sum_{c=1}^dK_c^{\alpha\beta}(t)\mu_c^\beta(t),
\label{eq:beff}
\end{equation}
and the L0 approximation of Eq.~\eqref{eq:exact-elimination} therefore becomes
\begin{equation}
\cA_{\rm L0}(T)=\left(\prod_c\tau_c\right)
\left\langle \bn_0 \,\middle|\,
\cT\exp\left[-\ii\int_0^T \bB_{\rm eff}(t)\cdot\bS_0\,\dd t\right]
\,\middle|\, \bn_0\right\rangle.
\label{eq:L0-general}
\end{equation}
This is still a fully quantum one-spin boundary-value problem.  Its field may be complex because $\mu_c(t)$ is a weak value, not an ordinary expectation value.

To obtain G1, keep the quadratic term in \eqref{eq:W-cumulant} and evaluate its first contribution to the logarithm of the hub L0 amplitude.  Let $U_{0,{\rm L0}}$ be the L0 hub propagator and let
\begin{equation}
G_0^{\alpha\gamma}(t,t')=
\frac{\left\langle\bn_0\,\middle|\,\cT[S_{0,I}^\alpha(t)S_{0,I}^\gamma(t')]\,U_{0,{\rm L0}}(T,0)\,\middle|\,\bn_0\right\rangle}
{\left\langle\bn_0\,\middle|\,U_{0,{\rm L0}}(T,0)\,\middle|\,\bn_0\right\rangle}
\label{eq:G0-general}
\end{equation}
be the full weak hub two-point function under the L0 drive.  The first Gaussian correction to $\log\cA$ is then
\begin{equation}
\boxed{
\Delta_{\rm G1}(T)=
-\frac12\sum_{c=1}^d\int_0^T\!\int_0^T
K_c^{\alpha\beta}(t)K_c^{\gamma\delta}(t')
\kappa_c^{\beta\delta}(t,t')G_0^{\alpha\gamma}(t,t')\,\dd t\,\dd t'.}
\label{eq:G1-general}
\end{equation}
The summation convention over repeated Cartesian indices is used throughout.

The $1/d$ ordering of the hierarchy can now be stated precisely.  We formulate it as an asymptotic ordering of the cumulant series of Lemma~\ref{lem:elimination}, under explicit hypotheses, and we are careful to distinguish what is proved (the ordering of the series, term by term) from what is validated numerically (that the exact remainders follow the same orders).

\begin{proposition}[Ordering of the cumulant series]
\label{prop:remainder}
Assume, on the fixed window $[0,T]$:
\begin{enumerate}[label={\rm(H\arabic*)},itemsep=1pt,topsep=2pt]
\item bounded drives: $\bB_0$, $\bb_c$, $K_c$ are piecewise continuous with $\|K_c(t)\|\le C_K/d$ for all $c$ and $t$;
\item no weak-amplitude zeros: along the continuous branch from $\cA(0)=1$, the free leaf amplitudes and the L0 hub amplitude satisfy $|\tau_c|,\,|\langle\bn_0|U_{0,\rm L0}(T,0)|\bn_0\rangle|\ge\delta>0$, and the same lower bound holds for the intermediate boundary amplitudes defining the weak objects;
\item uniform weak-cumulant bound: the connected $k$-point weak cumulants of each leaf obey $\big|C_c^{\beta_1\cdots\beta_k}(t_1,\ldots,t_k)\big|\le k!\,M^k$ for a constant $M=M(\delta)$ independent of $c$, $k$, and $d$.
\end{enumerate}
Then the term of the expansion of $\log\cA$ generated by total leaf-cumulant order $k\ge2$ is bounded by $d\,(C_K T M/d)^k\,c_0^k=O(d^{1-k})$ with a constant $c_0$ depending only on the hub norm $\|\bS_0\|=\tfrac12$ and on $\delta$.  In particular the first term omitted by L0 is $O(1/d)$, the first term omitted by ${\rm L0{+}G1}$ -- comprising the cubic leaf cumulants and the second-order contribution of the quadratic insertion, $\log\langle e^{Q_2}\rangle_{\rm L0}-\langle Q_2\rangle_{\rm L0}=O(Q_2^2)$ -- is $O(1/d^2)$, and
\begin{equation}
\log\cA(T)-\log\cA_{\rm L0}(T)=O(1/d),
\qquad
\log\cA(T)-\log\cA_{\rm L0}(T)-\Delta_{\rm G1}(T)=O(1/d^2),
\label{eq:error-scaling}
\end{equation}
as an asymptotic ordering of the series, provided the branch of the logarithm is followed continuously from $\cA(0)=1$.
\end{proposition}

\begin{proof}[Proof sketch]
For real drives the leaf propagators are unitary, so every $k$-point weak moment is a ratio of a matrix element of a product of $k$ operators of norm $\tfrac12$, sandwiched between unit vectors and unitaries, to an amplitude bounded below by $\delta$ under (H2); hence weak moments are bounded by $(1/2)^k/\delta$.  Connected cumulants are finite signed sums of products of moments over partitions, giving the factorially weighted bound (H3) with $M=M(\delta)$; the factorial is compensated by the $1/k!$ of the exponential series and the time-ordered simplex volume $T^k/k!$.  A term of total leaf-cumulant order $k$ carries $k$ powers of $K_c$ for the same leaf, i.e., a factor $\le(C_K/d)^k$, hub insertions of norm $\tfrac12$ evaluated in the weak L0 state (bounded using (H2)), and a sum over $d$ leaves.  Collecting factors gives the stated $O(d^{1-k})$ bound; Appendix~\ref{app:cumulant-derivation} records the same counting in coherent-state language, including the $O(Q_2^2)=O(d^{-2})$ estimate for the re-exponentiation error of the quadratic insertion.
\end{proof}

\begin{remark}
Proposition~\ref{prop:remainder} orders the terms of the cumulant series; it does not by itself control the resummation of the full tail, so Eq.~\eqref{eq:error-scaling} for the \emph{exact} remainders is an asymptotic statement supported here by direct numerical validation: the homogeneous scaling test of Fig.~\ref{fig:one-over-d} (fitted slopes $-1.00$ and $-2.00$ over $d\in[8,96]$) and the fully driven nested-ensemble test of Fig.~\ref{fig:driven-scaling-ensemble} (ensemble slopes $-1.05$ and $-2.03$; per-seed slopes $-1.02\pm0.26$ and $-2.00\pm0.37$).  The physical content of hypothesis (H2) -- its failure at Fisher zeros of dynamical quantum phase transitions -- was discussed in Sec.~\ref{sec:general-star}.
\end{remark}  
All numerical sections below compare exact amplitudes against Eqs.~\eqref{eq:L0-general} and \eqref{eq:G1-general}, using the branch-continuous logarithmic error natural for this cumulant hierarchy.  Fig.~\ref{fig:theory-schematic} summarizes the two levels of the construction and the one-spin objects entering each level.

\begin{figure}[H]
\centering
\includegraphics[width=0.82\textwidth]{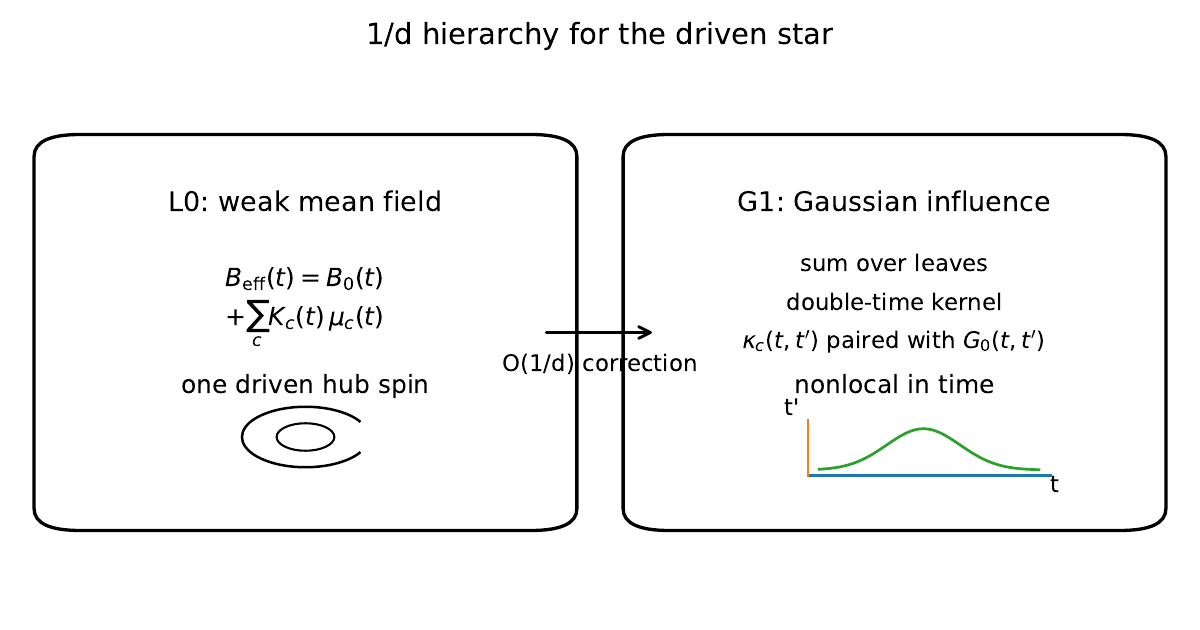}
\caption{The $1/d$ hierarchy. L0 is a mean weak field. G1 is the first Gaussian influence correction, nonlocal in time but still assembled from one-spin objects.}
\label{fig:theory-schematic}
\end{figure}

\section{Special cases and exact validation rungs}
\label{sec:special-cases}

The driven formulation above is the main theory. The static and homogeneous models are best presented as validation rungs: they give exact or semi-exact oracles against which the $1/d$ hierarchy can be tested.

\subsection{Static homogeneous star}
\label{sec:homogeneous-star}

Set $\bB_0=\bb_c=0$ and $K_c(t)=(J_0/d)I$. Then
\begin{equation}
H=\frac{J_0}{d}\sum_{c=1}^d\bS_0\cdot\bS_c.
\label{eq:homogeneous-H}
\end{equation}
Appendix~\ref{app:static-reduction} shows how Eqs.~\eqref{eq:L0-general}--\eqref{eq:G1-general} reduce to the static homogeneous formulas used in this validation rung.
For aligned leaves and a hub at polar angle $\theta$ relative to the leaf direction, the exact amplitude reduces to a two-sector interference problem with revival scale
\begin{equation}
T_{\rm rev}=\frac{4\pi d}{J_0(d+1)}.
\label{eq:revival}
\end{equation}
This case is ideal for large-$d$ scaling tests because exact amplitudes can be evaluated at very large $d$ without exponential Hilbert-space growth; Fig.~\ref{fig:aligned-benchmark} shows the corresponding exact benchmark, and the closed form is recorded in Appendix~\ref{app:oracles}.

\begin{figure}[H]
\centering
\includegraphics[width=0.82\textwidth]{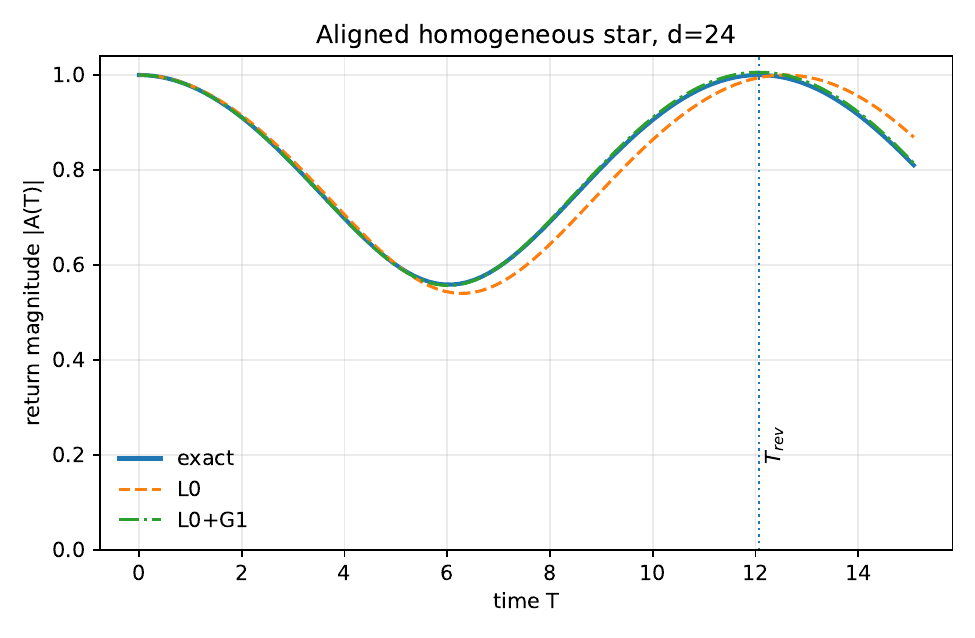}
\caption{The aligned homogeneous star gives an exact benchmark and exposes the finite-time horizon imposed by coherent revivals.}
\label{fig:aligned-benchmark}
\end{figure}

\subsection{Homogeneous arbitrary coherent states}
\label{sec:homogeneous-coherent}

The homogeneous star is more solvable than the aligned example suggests; Appendix~\ref{app:oracles} gives the corresponding Schur-block recursion. Since $H=(J_0/d)\bS_0\cdot\bm L$ with $\bm L=\sum_c\bS_c$, each total-leaf-spin sector has only two hub-leaf total-spin eigenvalues. This yields an operator identity of the form
\begin{equation}
\exp[-\ii (J_0T/d)\bS_0\cdot\bm L]=f_0(\bm L^2)+f_1(\bm L^2)\bS_0\cdot\bm L,
\label{eq:homogeneous-operator-identity}
\end{equation}
combined with a Schur-block recursion over the leaf product state. This gives a polynomial-time exact oracle for arbitrary leaf coherent states.

\subsection{Inhomogeneous static star}

Set fields and couplings time independent but allow arbitrary $J_c$, $\bb_c$, and $\bB_0$. This is no longer generally reducible to the homogeneous Schur-block oracle. Nevertheless, it is a direct specialization of the driven theory, with the one-spin propagators \eqref{eq:leaf-prop} computable analytically or numerically. It is the right intermediate validation between the exactly solvable homogeneous rung and the fully driven case.

A further middle rung is the Gaudin central-spin model, with inhomogeneous exchange and hub field but no arbitrary leaf fields. This model has integrable structure \cite{gaudin_diagonalisation_1976,bortz_exact_2007}; it may provide a large-$d$ oracle beyond brute-force exact diagonalization.

\section{Numerical validation}
\label{sec:numerics}

The presented figure set separates the philosophical benchmark from the algorithmic validation: Fig.~\ref{fig:uniform-ll-to-quantum} is the LL-to-quantum diagnostic, Fig.~\ref{fig:star-primitive} shows the influence primitive, Fig.~\ref{fig:one-over-d} tests the headline $1/d$ and $1/d^2$ scaling, Fig.~\ref{fig:static-inhom} probes a static inhomogeneous case beyond the homogeneous oracle, Figs.~\ref{fig:driven-validation}--\ref{fig:driven-scaling-ensemble} validate the fully time-dependent algorithmic setting, and Fig.~\ref{fig:tmps-baseline} benchmarks the hierarchy against a temporal matrix-product influence-matrix baseline.

\begin{figure}[H]
\centering
\includegraphics[width=0.82\textwidth]{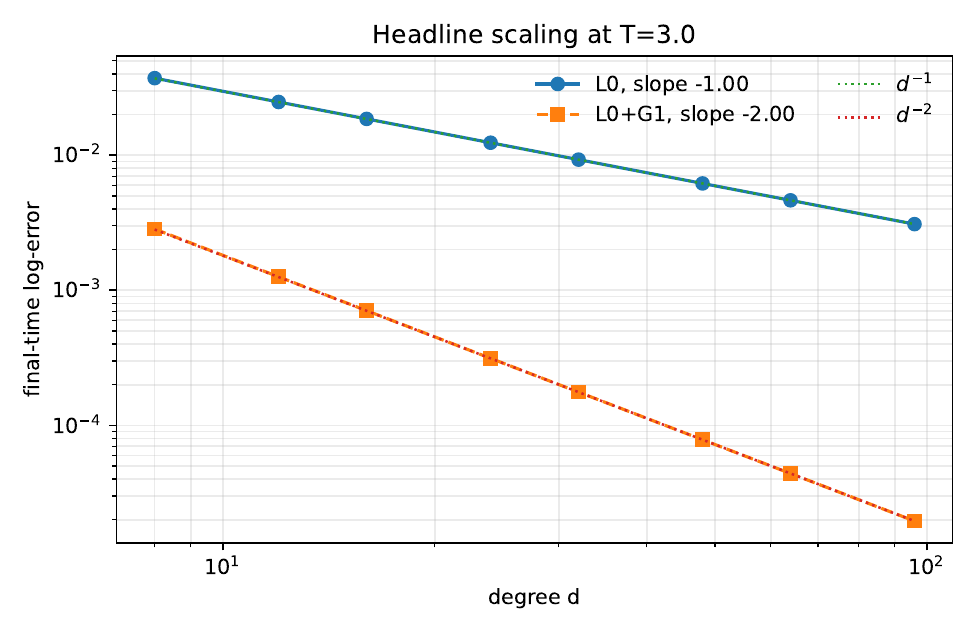}
\caption{Central scaling test. The leading error is $O(1/d)$, and the G1-corrected error is $O(1/d^2)$.}
\label{fig:one-over-d}
\end{figure}

\begin{figure}[H]
\centering
\includegraphics[width=0.82\textwidth]{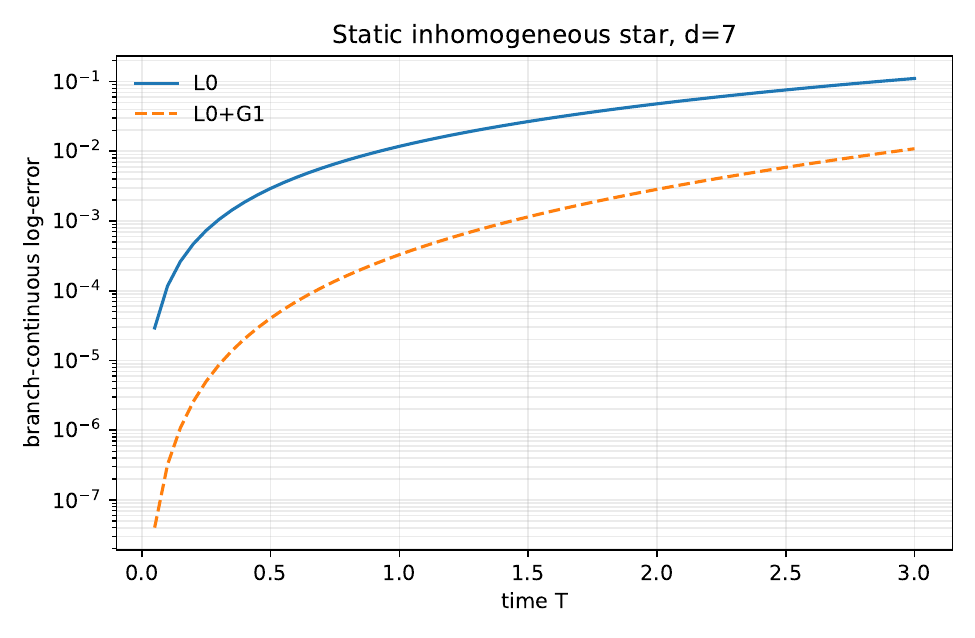}
\caption{Static inhomogeneous validation tests the hierarchy beyond the homogeneous exact oracle.  At late times, outside the asymptotic window controlled by Proposition~\ref{prop:remainder}, the L0+G1 error can cross above the L0 error; the hierarchy is an ordering of small corrections, not a variational bound.}
\label{fig:static-inhom}
\end{figure}

\begin{figure}[H]
\centering
\includegraphics[width=0.82\textwidth]{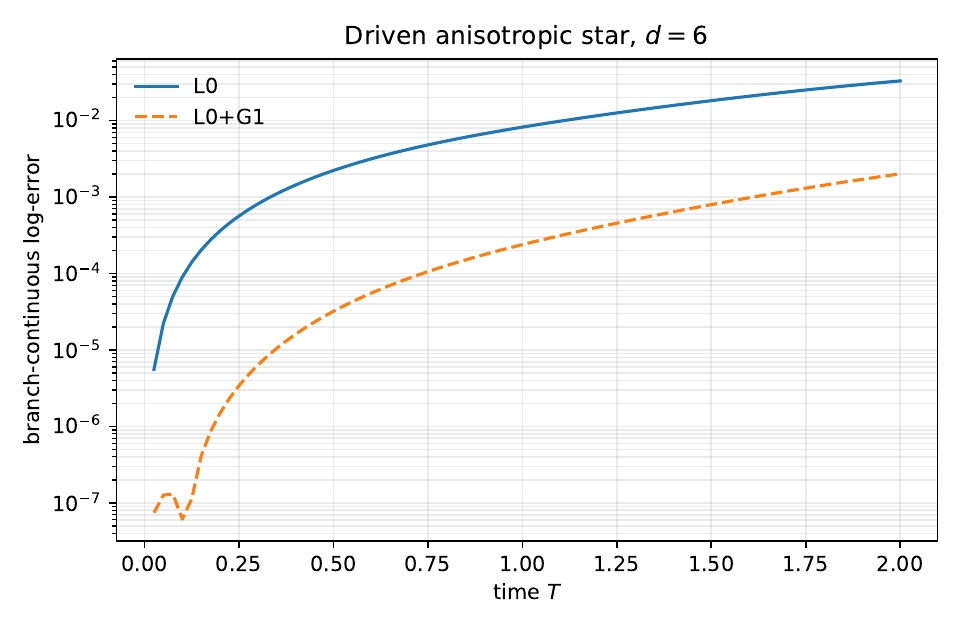}
\caption{The driven-star experiment validates the fully time-dependent formulation directly, using smooth time-dependent fields and anisotropic pair tensors.}
\label{fig:driven-validation}
\end{figure}

\subsection{Driven-star validation with nested ensembles}
\label{sec:driven-ensemble}

The fully driven test uses smooth time-dependent hub fields, leaf fields, and anisotropic pair tensors, together with a piecewise-constant midpoint exact oracle at small $d$.  Its purpose is algorithmic rather than integrable: it verifies that the general formulas \eqref{eq:L0-general}--\eqref{eq:G1-general} are operational when $\bB_0(t)$, $\bb_c(t)$, and $K_c(t)$ are genuinely time dependent.  Fig.~\ref{fig:driven-validation} shows the time-resolved log-errors for one instance.

Because the exact oracle limits the driven test to small degrees, slope fits over $d\in[4,12]$ are sensitive to instance-to-instance prefactor fluctuations if every degree draws an independent random Hamiltonian.  The scaling test therefore uses \emph{nested ensembles}: each leaf $c$ carries intrinsic parameters (orientation, field waveform, anisotropic pair-tensor pattern) drawn from a per-leaf random stream that depends only on the ensemble seed and on $c$, so that increasing $d$ \emph{adds} leaves without redrawing the existing ones, and $d$ enters only through the overall $1/d$ coupling normalization.  Fig.~\ref{fig:driven-scaling-ensemble} shows the final-time log-errors averaged over eight such ensembles: the ensemble-mean slopes are $-1.05$ for L0 and $-2.03$ for L0+G1, with per-seed slopes $-1.02\pm0.26$ and $-2.00\pm0.37$.  This is the fully driven, anisotropic counterpart of the homogeneous large-$d$ test of Fig.~\ref{fig:one-over-d}, and it directly supports the orders of Proposition~\ref{prop:remainder} in the algorithmic setting.  Appendix~\ref{app:numerics} describes the branch-continuous log-error, midpoint propagation, and ordered-triangle quadrature.

\begin{figure}[H]
\centering
\includegraphics[width=0.72\textwidth]{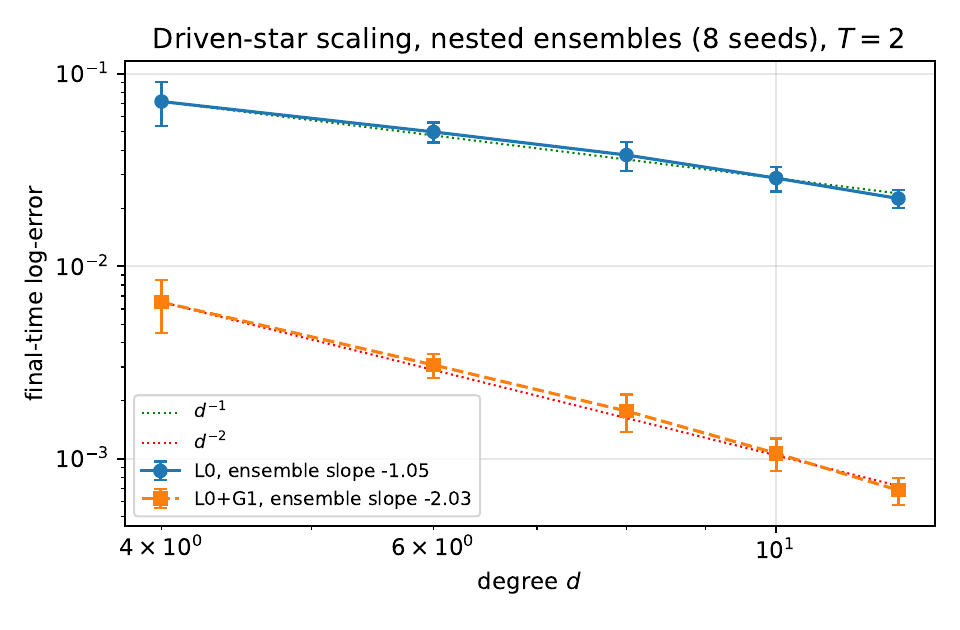}
\caption{Fully driven anisotropic star: final-time log-errors over nested ensembles (eight seeds; markers and bars show ensemble mean and standard deviation).  Increasing $d$ adds leaves to a fixed per-leaf random instance, so the fits are not contaminated by prefactor redraws.  The ensemble slopes $-1.05$ (L0) and $-2.03$ (L0+G1) match the $d^{-1}$ and $d^{-2}$ guides.}
\label{fig:driven-scaling-ensemble}
\end{figure}

\subsection{Comparison with a temporal-MPS influence-matrix baseline}
\label{sec:tmps-baseline}

The elimination-then-compression logic of this paper is shared by a mature family of tensor-network methods, and the usefulness of the Gaussian message must be calibrated against the strongest of them.  In the influence-matrix (IM) approach \cite{banuls_matrix_2009,strathearn_efficient_2018,lerose_influence_2021,sonner_influence_2021}, the effect of an eliminated environment on a retained system is encoded, after discretization of time into $n_s$ slices, in a tensor carrying one index pair per time slice of the retained system -- the influence matrix, a discrete representation of the Feynman--Vernon influence functional -- and this tensor is stored and compressed as a matrix product along the \emph{time} direction.  Appendix~\ref{app:tmps} gives a self-consistent construction, including the definition of matrix-product states and operators and of the bond dimension; here we summarize the three steps and the single control parameter of the baseline.  (i) \emph{Discretization}: the propagator \eqref{eq:propagator} is factorized by a symmetric Trotter splitting \cite{trotter_product_1959,suzuki_generalized_1976} into $n_s$ slices of one-spin hub gates and two-spin hub--leaf gates.  (ii) \emph{Exact per-leaf elimination}: contracting the world line of one leaf against its boundary state $|\bnu_c\rangle$ converts it, with no further approximation, into a matrix-product operator (MPO) acting on the hub time slices, with bond dimension two -- the dimension of the leaf Hilbert space.  (iii) \emph{Aggregation with compression}: the influences of the $d$ leaves combine by slice-wise multiplication of their MPOs, so the bond dimension of the exact combined influence grows up to $2^d$; whenever the running bond exceeds a preset maximum $\chi$, every bond is truncated to its $\chi$ largest singular values by the standard sweep algorithm recalled in Appendix~\ref{app:tmps}.  The bond dimension $\chi\in\{1,2,3,\dots\}$ is thus the single accuracy knob of the baseline: $\chi=2^d$ reproduces the discretized dynamics exactly, and decreasing $\chi$ trades accuracy for cost.  Contracting the compressed influence with the hub gates and the hub boundary states yields the IM estimate of $\cA(T)$.  Applying the identical gate sequence to the full $2^{d+1}$-dimensional statevector gives a circuit-exact reference, so the compression error can be isolated from the Trotter (discretization) error; the implementation and its unit tests are released with the code.

Fig.~\ref{fig:tmps-baseline} reports the comparison; three findings emerge.  Panel~(a) compares accuracy against message size at $d=8$, $T=2$, and reference coupling scale $g=1$.  Here and below, the dimensionless coupling scale $g$ multiplies all pair tensors of the instance uniformly, $K_c(t)\to g\,K_c(t)$, so that $\|K_c\|=O(g/d)$; the nested-ensemble validation instances of Sec.~\ref{sec:driven-ensemble} correspond to $g=1$, and panel~(c) scans $g$ at fixed $d$.  Message size counts the complex parameters of the data structure that each method delivers to the hub -- the object that would be communicated in a graphical algorithm: $\sum_k 4\chi_k\chi_{k+1}$ for the IM message and $3n_t+9n_t^2$ for the aggregate Gaussian message $\big(\bB_{\rm eff},\,\sum_c K_c\kappa_cK_c^{\mathsf T}\big)$ stored on $n_t$ time grid points; both counts are derived in Appendix~\ref{app:tmps}.  The first finding is that at this small degree the IM baseline is decisively stronger: $\chi=1,2$ are crude, $\chi=3$ is comparable to L0+G1, and $\chi=4$ -- about $6\times10^{3}$ parameters -- already reaches the Trotter floor of the discretization ($\approx10^{-4}$ in the log-amplitude at $n_s=100$ slices), whereas L0+G1 reaches $1.4\times10^{-3}$ at $\sim10^{5}$ parameters ($n_t=101$; refining the grid to $n_t=201$ does not improve it).  The accuracy of the Gaussian message at $d=8$ is limited by its $O(1/d^2)$ truncation, not by its parameter count.

The two methods are, however, truncations of the same exact influence with respect to \emph{different} expansion parameters, and this makes their regimes complementary rather than ordered.  The Gaussian message truncates the cumulant series of Lemma~\ref{lem:elimination} at second order: its error falls as $1/d^2$ at fixed coupling (Proposition~\ref{prop:remainder}) and grows as $g^3$ -- the magnitude of the first omitted cumulant -- when the overall coupling scale $g$ is raised at fixed $d$.  The IM message truncates the bond dimension: it is exact at $\chi=2^d$ for \emph{any} coupling strength -- in this precise sense it is nonperturbative in the coupling -- and the $\chi$ required for a target accuracy is set by the temporal entanglement of the influence functional, which grows with the accumulated coupling $\|K\|T$ and is largely insensitive to $d$.  Panels~(b) and~(c) measure both statements.  In the coordination scan~(b), at fixed $g=1$ over nested ensembles, the L0+G1 error falls with fitted slope $-2.0$, while the fixed-$\chi$ IM errors decay far more slowly ($\chi=2$: slope $\approx-0.9$; $\chi=3$: nearly flat around $2\times10^{-3}$).  Consequently the Gaussian message is more accurate than $\chi=2$ throughout the scanned range, overtakes $\chi=3$ near $d\approx8$, and -- extrapolating its measured $d^{-2}$ law to the $\chi=4$ compression error ($7.7\times10^{-5}$ at $d=8$, panel~(a)) -- is projected to overtake $\chi=4$ near $d\approx35$; in the window shown, the total $\chi=4$ error tracks the common Trotter floor and is discretization- rather than compression-limited.  In the coupling scan~(c), at fixed $d=8$, the L0+G1 error follows the $\propto g^3$ guide while the fixed-$\chi$ IM errors grow more slowly.  Together the two scans trace the announced crossover surface in the $(d\|K\|T)$ plane, and yield the second finding: for every fixed IM budget $\chi$ there is a coordination $d_*(\chi)$, growing with $\chi$, beyond which the Gaussian message is the more accurate primitive at fixed, $d$-independent cost; conversely, at fixed degree a sufficiently strong coupling always favors the IM message.

The high-coordination side of this complementarity has a transparent origin -- a central-limit mechanism.  As $d\to\infty$ with $K_c=O(1/d)$, the exact combined influence of the leaves approaches a \emph{Gaussian} functional: its mean part is $O(1)$, its two-time kernel is $O(1/d)$, and all higher connected contributions are $O(1/d^2)$ or smaller (Proposition~\ref{prop:remainder}).  This limiting object is exactly what the cumulant message stores in closed form, at a size independent of $d$; a fixed-rank temporal factorization must instead spend bond dimension on the nonlocal-in-time quadratic kernel, and in the measured window its error decays far more slowly than the $1/d^2$ of the Gaussian message.  The third finding is structural and matters most for graphical algorithms: the Gaussian message is \emph{additive} -- leaves and subtrees aggregate by summation of their $(\mu,\kappa)$ pairs, and the aggregate entering the hub is a single pair regardless of $d$ -- whereas IM messages compose multiplicatively, each additional member requiring an MPO multiplication followed by a fresh compression sweep.  On high-coordination trees this additivity is the structural reason to prefer the cumulant message; on small, strongly coupled clusters the IM message is the right primitive; hybrid schemes -- IM messages inside strongly coupled clusters, Gaussian messages between clusters and along high-degree links -- are the natural synthesis (Sec.~\ref{sec:outlook}).

\begin{figure}[H]
\centering
\includegraphics[width=0.98\textwidth]{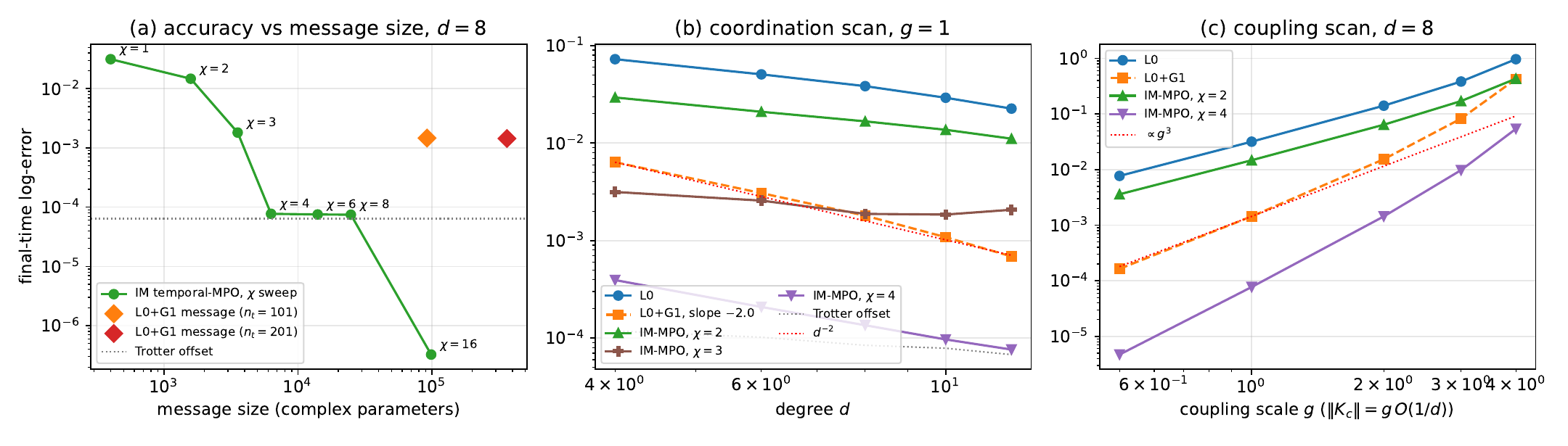}
\caption{Temporal-MPS influence-matrix (IM) baseline versus the Gaussian (L0+G1) message.  (a)~Final-time log-error against message size (complex parameters; Appendix~\ref{app:tmps}) at $d=8$, $T=2$, $g=1$.  The IM curve sweeps the bond dimension $\chi$ and is measured against the circuit-exact reference on the same Trotter grid ($n_s=100$), isolating the compression error; the dotted line marks the Trotter offset of the discretized dynamics from the continuum oracle; the L0+G1 diamonds (kernel stored on $n_t=101$ and $201$ grid points) are measured against the continuum oracle.  (b)~Coordination scan at fixed $g=1$: final-time log-errors against the common continuum oracle, averaged over four nested-ensemble seeds (Sec.~\ref{sec:driven-ensemble}).  L0+G1 falls with fitted slope $-2.0$; the fixed-$\chi$ IM errors decay slowly ($\chi=2$: slope $\approx-0.9$; $\chi=3$: nearly flat), so the Gaussian message overtakes $\chi=3$ near $d\approx8$, while the $\chi=4$ curve tracks the Trotter offset (dotted) and is discretization- rather than compression-limited in this window.  (c)~Coupling scan at fixed $d=8$: the L0+G1 error follows the $\propto g^{3}$ growth of the first omitted cumulant, while the fixed-$\chi$ IM errors grow more slowly.  Panels (b) and (c) together trace the crossover surface in the $(d\|K\|T)$ plane.}
\label{fig:tmps-baseline}
\end{figure}

\section{From stars to trees and loopy graphs}
\label{sec:trees-loops}

The star is the local update rule for a tree. A subtree rooted at $i$ sends to its parent $p$ an influence message
\begin{equation}
\mathfrak m_{i\to p}=\left(\ell_{i\to p},\;\mu_{i\to p}^\alpha(t),\;\kappa_{i\to p}^{\alpha\beta}(t,t')\right),
\label{eq:tree-message}
\end{equation}
where $\ell$ is the scalar log contribution, $\mu$ is the weak boundary-spin trajectory of the subtree root, and $\kappa$ is its connected two-time kernel. The update at a vertex is the star computation with the children playing the role of leaves. Concretely, the update at vertex $i$ with children $c\in\partial i\setminus p$ is the star computation of Secs.~\ref{sec:leaf-elimination}--\ref{sec:hierarchy} used in \emph{open} form.  The incoming children messages are aggregated into the dressed evolution of spin $i$: the weak mean trajectories shift the field of $i$ as in Eq.~\eqref{eq:beff}, and the children kernels enter as a quadratic, nonlocal-in-time insertion treated at first order, consistently with the Gaussian truncation.  The outgoing message $\mathfrak m_{i\to p}$ then consists of the same one-spin weak objects, Eqs.~\eqref{eq:tau}--\eqref{eq:kappa-general}, evaluated for spin $i$ under this dressed evolution with $i$'s own coherent boundary state $|\bnu_i\rangle$: $\ell_{i\to p}$ collects the log of the dressed return amplitude of $i$ together with the children's scalars, and $\mu_{i\to p}$, $\kappa_{i\to p}$ are the dressed weak trajectory and connected kernel.  Rather than reporting the closed number $\cA$, the vertex reports its influence as felt by the parent; the closure $\langle\bn_0|\cdots|\bn_0\rangle$ of Eq.~\eqref{eq:return-amplitude} is performed once, at the tree root.  As emphasized in Remark~\ref{rem:scope}(ii), the elimination itself remains exact at every generation, but a finite message requires truncating the cumulant functional at $(\mu,\kappa)$ each time it is passed, so beyond the single star the recursion is approximate by construction, with the per-generation error ordered by Proposition~\ref{prop:remainder} when coordination is high. For high-degree trees with $K_{ij}=O(1/d_i)$, the same cumulant counting applies generation by generation.

For bounded-degree trees, the algorithmic message-passing formulation still exists, but $1/d$ is no longer a small parameter unless replaced by weak coupling, short time, high temperature, dephasing, or another expansion parameter.

For loopy graphs, the natural extension is tree/star Gaussian message passing plus loop corrections. This connects directly to belief propagation on graphical models, generalized belief propagation, tensor-network contraction, and loop calculus \cite{bethe_statistical_1935,pearl_probabilistic_1988,yedidia_constructing_2003,wainwright_graphical_2008,chertkov_loop_2006,chertkov_loop_2006-1}. The real-time quantum setting is harder because amplitudes are complex rather than positive, so this direction should be framed as a separate research program.

\section{Discussion and outlook}
\label{sec:outlook}

The paper has two complementary messages:  
\begin{enumerate}
\item \underline{\emph{Quantum Advantage.}}  A qubit spin system is not automatically a useful quantum-simulation benchmark merely because its microscopic degrees of freedom are quantum.  For a specified observable and time window one should first ask whether a spin-LL-classical surrogate reproduces the qualitative behavior.  In the present hierarchy this surrogate is L0.  If L0 fails but L0+G1 succeeds, the dynamics is not fully classical but remains classically correctable by a controlled Gaussian quantum-sector term.  If both fail, the observable has entered a genuinely quantum regime relative to this diagnostic.  The uniform-star experiment in Fig.~\ref{fig:uniform-ll-to-quantum} is the cleanest current illustration: the homogeneous model is exactly solvable by Schur-block recursion, yet the two-population leaf state generates enough sector structure to reveal an eventual departure from the LL description on the scale where discrete-spectrum interference becomes important.

\item \underline{\emph{Algorithmic Message.}}  The fully time-dependent star in Eq.~\eqref{eq:general-H} is more important for scalable methods than the uniform benchmark.  It allows arbitrary bounded hub fields, leaf fields, and anisotropic pair tensors, and it returns precisely the objects needed for graphical propagation: a scalar contribution, a weak mean trajectory, and a connected two-time kernel.  Thus the star is a local influence primitive.  On high-degree trees, applying this primitive recursively gives a Gaussian influence-message scheme with the same $1/d$ counting generation by generation.  On bounded-degree trees the same recursion remains meaningful, but accuracy must be controlled by a different small parameter such as weak coupling, dephasing, high temperature, or a short-time expansion.

The relation to \emph{tensor networks} and \emph{graphical models} is also structural.  In classical graphical models, BP is exact on trees and approximate on loopy graphs; loop calculus expresses the exact partition function as the BP contribution plus closed-loop corrections \cite{chertkov_loop_2006,chertkov_loop_2006-1,chertkov_gauges_2020,chertkov_inferlo_2024}.  Although this construction was developed in the language of probabilities and statistical weights, its gauge-transformation core is algebraic: it reparametrizes a finite tensor contraction without changing the contracted value.  In this sense the same logic is not restricted to positive measures, and extends naturally to graphical models with real, signed, or complex entries.  Recent tensor-network work applies this BP-plus-corrections viewpoint directly to complex-amplitude contractions \cite{evenbly_loop_2026,midha_beyond_2025,midha_belief_2026}.  Our real-time spin setting adds a further layer: the messages are not finite local tables only, but time-continuous weak objects, namely weak mean trajectories and two-time weak kernels.  Nevertheless, the same organizing principle suggests a route beyond trees: start from the star/tree Gaussian message approximation, then add loop or cluster corrections when closed graphical structures are important.
\end{enumerate}

Several limitations should be kept explicit:
\begin{enumerate}
    \item  \underline{\emph{Classical substitutability}} is not an intrinsic yes/no property of a Hamiltonian.  It depends on the observable, the initial state, the time horizon, and the accuracy criterion.  Return amplitudes are phase-sensitive and therefore more demanding than low-order equal-time observables.  The branch-continuous logarithm can become singular near zeros of the weak return amplitude; since the return amplitude is a Loschmidt amplitude, these zeros are the Fisher zeros of dynamical quantum phase transitions \cite{heyl_dynamical_2013,heyl_dynamical_2018}, so hypothesis (H2) of Proposition~\ref{prop:remainder} delineates the physically expected domain of validity.  
    
    \item \underline{\emph{Baseline comparison with Tensor Network}} discussed in Sec.~\ref{sec:tmps-baseline} shows that at small degree and moderate coupling the temporal influence-matrix message is the stronger classical primitive; the case for the Gaussian message rests on its fixed size, its additivity across leaves and subtrees, and its $1/d^2$-controlled accuracy at high coordination; the case is made quantitative in Sec.~\ref{sec:tmps-baseline} and distilled into the selection rule below.

\item \underline{\emph{Focus on Methodology.}} The statement that a qubit scheme can be tested against a spin-LL surrogate is a methodology, not a universal theorem: the surrogate can be derived broadly, but perturbative control must be established case by case.
\end{enumerate} 

The baseline comparison also yields the clearest \underline{algorithmic message of the paper} -- \emph{a quantitative selection rule for real-time spin message passing}.  The two message types truncate the same exact influence in different expansion parameters: influence matrices in temporal entanglement, which grows with the accumulated coupling $\|K\|T$ and is insensitive to coordination; Gaussian cumulant messages in $1/d$, insensitive to bond-dimension-type resources but perturbative in the coupling.  Accordingly, influence-matrix messages should be used inside small or strongly coupled clusters, where non-perturbative accuracy in the coupling is required; Gaussian messages should be used across high-coordination, weakly coupled links, where their error falls as $1/d^2$ at fixed, additive cost while fixed-$\chi$ influence matrices stagnate; and the measured crossover surface $d_*(\chi)$ of Fig.~\ref{fig:tmps-baseline} locates the boundary between the two regimes (the Gaussian message is more accurate than bond $\chi=2$ throughout the scanned range, overtakes $\chi=3$ near $d\approx8$, and is projected to overtake $\chi=4$ near $d\approx35$; at fixed degree, a sufficiently strong coupling always favors the influence matrix).  Table~\ref{tab:message-comparison} summarizes the structural side of the same comparison.  The synthesis -- hybrid message passing with influence-matrix messages within clusters and Gaussian messages between them -- is in our view the most promising purely algorithmic continuation of Sec.~\ref{sec:trees-loops}.

The next scientific steps are therefore clear: 
\begin{enumerate}
\item Sharpen the uniform-star LL-to-quantum transition into a robust diagnostic figure, possibly with disorder ensembles and alternative initial leaf distributions.

\item Use the fully driven star to build and test an explicit tree-message prototype.

\item Compare the resulting tree approximation against temporal-MPS influence-matrix and exact small-system oracles, extending the star-level comparison of Sec.~\ref{sec:tmps-baseline} to the tree level, where the additivity of the Gaussian message and the bond growth of multiplied influence matrices compete directly; the selection rule above, calibrated by the measured crossover surface, tells the prototype when to switch message types.
\end{enumerate} 

\section*{Data and Code Availability}

The source code used to generate the numerical results in this work, together with the scripts that produce the figures and the data files generated by those scripts, are available in a public GitHub repository at \url{https://github.com/mchertkov/QuantumStar}.  The repository includes notebooks reproducing every figure.

\section*{Funding}

The author acknowledges partial support from University of Arizona start-up funds and from a fellowship of the Alexander von Humboldt Foundation.

\section*{Acknowledgments}

The author thanks I. Klich for useful discussions and literature pointers.

\section*{Use of AI-Assisted Tools}

Large language models (Claude, Anthropic; ChatGPT, OpenAI) assisted with text editing and code refactoring; all mathematical derivations, scientific claims, and code were independently verified by the author.


\appendix

\section{One-spin weak objects}
\label{app:weak-objects}

All ingredients of the star hierarchy reduce to boundary-value objects of a single driven spin.  Let
\begin{equation}
P(t)=U(t,0),
\qquad
P(T)=U(T,0),
\qquad
U(T,t)=P(T)P(t)^{-1},
\end{equation}
where $U$ is generated by a possibly complex field $\bb(t)$ through $\dot U=-\ii\,\bb(t)\cdot\bS\,U$.  For a coherent boundary state $|\psi\rangle=|\bnu\rangle$, define
\begin{equation}
\tau=\langle\psi|P(T)|\psi\rangle.
\end{equation}
The weak one-point trajectory is
\begin{equation}
\mu^\alpha(t)=
\frac{\left\langle\psi\,\middle|\,P(T)P(t)^{-1}S^\alpha P(t)\,\middle|\,\psi\right\rangle}{\tau}.
\label{eq:app-mu}
\end{equation}
For $t_i\ge t_j$, the full time-ordered weak two-point function is
\begin{equation}
G_>^{\alpha\beta}(t_i,t_j)=
\frac{\left\langle\psi\,\middle|\,P(T)P(t_i)^{-1}S^\alpha P(t_i)P(t_j)^{-1}S^\beta P(t_j)
\,\middle|\,\psi\right\rangle}{\tau}.
\label{eq:app-Ggt}
\end{equation}
The connected kernel entering the influence functional is
\begin{equation}
\kappa_>^{\alpha\beta}(t_i,t_j)=G_>^{\alpha\beta}(t_i,t_j)-\mu^\alpha(t_i)\mu^\beta(t_j),
\qquad t_i\ge t_j,
\label{eq:app-kappa-gt}
\end{equation}
and the $t_i<t_j$ branch is obtained by time ordering, i.e. by interchanging both the operator labels and the time arguments.  Numerically, the code stores the forward propagators $P(t_i)$ on a grid and constructs the inverse matrices $P(t_i)^{-1}$.  Since these are $2\times2$ matrices, the cost of producing all one- and two-time objects is $O(n_t^2)$ per leaf and is negligible compared with exact many-spin evolution.

In the static zero-field leaf case, $P(t)=I$ and $\tau=1$.  For a spin coherent state polarized along $\bnu$,
\begin{equation}
\mu^\alpha(t)=\frac{\nu^\alpha}{2},
\end{equation}
and the connected kernel reduces to
\begin{equation}
\kappa^{\alpha\beta}(t,t')=\frac14\left(\delta^{\alpha\beta}-\nu^\alpha\nu^\beta\right)
+\frac{\ii}{4}\varepsilon^{\alpha\beta\gamma}\nu^\gamma\,\mathrm{sgn}(t-t'),
\end{equation}
which is Eq.~\eqref{eq:static-kappa}.  The real symmetric part is the coherent-state covariance; the imaginary antisymmetric part is the response carried by spin commutation relations.

\section{Cumulant and coherent-state derivation of the hierarchy}
\label{app:cumulant-derivation}

This appendix records the derivation behind Sec.~\ref{sec:hierarchy}.  Write the Hamiltonian as
\begin{equation}
H(t)=\bB_0(t)\cdot\bS_0+H_{\rm leaf}(t)+V(t),
\qquad
H_{\rm leaf}(t)=\sum_c\bb_c(t)\cdot\bS_c,
\qquad
V(t)=\sum_c S_0^\alpha K_c^{\alpha\beta}(t)S_c^\beta.
\end{equation}
In the interaction picture with respect to the free leaves, the $c$-th leaf contribution is the time-ordered generating functional
\begin{equation}
Z_c[h]=\left\langle\bnu_c\,\middle|\,
\cT\exp\left[-\ii\int_0^T h^\beta(t)S_{c,I}^\beta(t)\,\dd t\right]
U_c(T,0)
\,\middle|\,\bnu_c\right\rangle,
\qquad
h^\beta(t)=K_c^{\alpha\beta}(t)S_0^\alpha(t).
\end{equation}
Its normalized logarithm $W_c[h]=\log Z_c[h]-\log Z_c[0]$ is the connected weak cumulant generator.  Functional differentiation gives
\begin{align}
\frac{\delta W_c}{\delta[-\ii h^\beta(t)]}\bigg|_{h=0}&=\mu_c^\beta(t),\\
\frac{\delta^2 W_c}{\delta[-\ii h^\beta(t)]\delta[-\ii h^\delta(t')]}\bigg|_{h=0}&=\kappa_c^{\beta\delta}(t,t'),
\end{align}
with $\mu_c$ and $\kappa_c$ defined in Eqs.~\eqref{eq:mu}--\eqref{eq:kappa-general}.  Hence
\begin{equation}
W_c[h]= -\ii\int h\cdot\mu_c
-\frac12\int\!\int h^\beta(t)h^\delta(t')\kappa_c^{\beta\delta}(t,t')\,\dd t\,\dd t'
+W_c^{(\ge3)}[h].
\end{equation}
Substituting the operator-valued source $h_c=K_c^{\mathsf T}S_0$ and keeping the linear term exactly gives the L0 hub problem.  The quadratic insertion is
\begin{equation}
Q_2=-\frac12\sum_c\int\!\int
K_c^{\alpha\beta}(t)K_c^{\gamma\delta}(t')
\kappa_c^{\beta\delta}(t,t')
S_0^\alpha(t)S_0^\gamma(t')\,\dd t\,\dd t'.
\end{equation}
Taking its weak L0 expectation gives Eq.~\eqref{eq:G1-general}.  Corrections from $\log\langle e^{Q_2}\rangle_{\rm L0}-\langle Q_2\rangle_{\rm L0}$ are $O(Q_2^2)=O(d^{-2})$, the same order as the cubic leaf cumulants, and are therefore beyond G1.

The coherent-state representation gives the same construction in a more geometric language.  Inserting spin coherent-state resolutions of the identity yields schematically
\begin{equation}
\cA(T)=\int_{\bn_i(0)=\bn_i(T)}\prod_i {\cal D}\bn_i\;
\exp\left(\ii\sum_i S_{\rm B}[\bn_i]-\ii\int_0^T H_{\rm cl}(\bn_0,\ldots,\bn_d;t)\,\dd t\right),
\end{equation}
where $S_{\rm B}$ is the Berry-phase term and $H_{\rm cl}$ is the coherent-state symbol of the Hamiltonian.  Integrating out the leaf paths produces an effective hub action
\begin{equation}
S_{\rm eff}[\bn_0]=S_{\rm B}[\bn_0]-\int_0^T\bB_0(t)\cdot\bs_0(t)\,\dd t
+\sum_c W_c[K_c^{\mathsf T}\bs_0],
\end{equation}
with $\bs_0$ the coherent-state spin vector on the hub path.  The L0 saddle obeys the complex boundary-value Landau--Lifshitz equation
\begin{equation}
\dot{\bs}_0(t)=\bB_{\rm eff}(t)\times\bs_0(t),
\end{equation}
while retaining the exact $2\times2$ hub propagator in the main text is the spin-$1/2$ version of the same L0 sector.  The G1 term is the first connected fluctuation of the eliminated leaf paths.  Thus the LL language and the operator weak-value language are two representations of the same hierarchy.

\section{Reduction to the static homogeneous formulas}
\label{app:static-reduction}

Set $\bB_0=\bb_c=0$ and $K_c(t)=(J_0/d)I$.  Then $\tau_c=1$ and $\mu_c=\bnu_c/2$.  The L0 field becomes
\begin{equation}
\bB_{\rm eff}=\frac{J_0}{2d}\sum_{c=1}^d \bnu_c\doteq  \bar{\bB},
\end{equation}
so that
\begin{equation}
\cA_{\rm L0}(T)=
\left\langle\bn_0\,\middle|\,\exp\left[-\ii T\,\bar{\bB}\cdot\bS_0\right] \,\middle|\,\bn_0\right\rangle.
\end{equation}
Writing $\bar B=|\bar{\bB}|$ and $\hat B=\bar{\bB}/\bar B$, this one-spin amplitude is
\begin{equation}
\cA_{\rm L0}(T)=
\cos\frac{\bar B T}{2}-\ii(\bn_0\cdot\hat B)\sin\frac{\bar B T}{2}.
\end{equation}
The G1 correction reduces to
\begin{equation}
\Delta_{\rm G1}(T)=
-\frac12\left(\frac{J_0}{d}\right)^2
\sum_{c=1}^d\int_0^T\!\int_0^T
\kappa_c^{\alpha\beta}(t,t')G_0^{\alpha\beta}(t,t')\,\dd t\,\dd t',
\end{equation}
with $\kappa_c$ given by Eq.~\eqref{eq:static-kappa} and $G_0$ computed under the constant L0 hub field $\bar{\bB}$.  This is the formula used in the homogeneous scaling and uniform-star LL-to-quantum experiments.

\section{Exact oracles}
\label{app:oracles}

\subsection{Aligned homogeneous star}

For aligned leaves the leaf system remains in the irreducible spin sector $L=d/2$.  Let
\begin{equation}
\lambda=\frac{J_0}{d},
\qquad
H=\lambda\,\bS_0\cdot\bm L,
\qquad
\bm L=\sum_{c=1}^d\bS_c.
\end{equation}
The two total-spin sectors $j=L\pm1/2$ have eigenvalues of $\bS_0\cdot\bm L$
\begin{equation}
x_+=\frac{L}{2},
\qquad
x_-=-\frac{L+1}{2},
\end{equation}
so the corresponding energies are $E_\pm=\lambda x_\pm$.  If the hub makes angle $\theta$ with the common leaf direction, define
\begin{equation}
c^2=\frac{1+\cos\theta}{2},
\qquad
s^2=\frac{1-\cos\theta}{2}.
\end{equation}
The exact return amplitude is
\begin{equation}
\cA_{\rm al}(T)=
c^2 e^{-\ii E_+T}
+s^2\left[\frac{1}{2L+1}e^{-\ii E_+T}+\frac{2L}{2L+1}e^{-\ii E_-T}\right].
\label{eq:app-aligned-amplitude}
\end{equation}
The relative phase has frequency
\begin{equation}
E_+-E_- = \frac{J_0(d+1)}{2d},
\end{equation}
which gives the revival scale quoted in the main text,
\begin{equation}
T_{\rm rev}=\frac{4\pi d}{J_0(d+1)}.
\end{equation}

\subsection{Schur-block oracle for arbitrary leaf coherent states}

For arbitrary leaf coherent states the homogeneous star is still polynomially solvable.  The leaf product density matrix is decomposed into Schur blocks $M_L$, one block for each total leaf spin $L$, with multiplicities already traced out.  The recursion starts from the first leaf,
\begin{equation}
M_{1/2}^{(1)}=|\bnu_1\rangle\langle\bnu_1|.
\end{equation}
When adding leaf $m+1$, one combines every existing block with the next spin-$1/2$ density matrix and projects onto the allowed sectors $L'=L\pm1/2$:
\begin{equation}
M_{L'}^{(m+1)}=
\sum_{L:\,L'=L\pm 1/2}
C_{L\to L'}\left(M_L^{(m)}\otimes |\bnu_{m+1}\rangle\langle\bnu_{m+1}|\right)C_{L\to L'}^\dagger.
\label{eq:app-schur-recursion}
\end{equation}
Terms with inadmissible $L'$ are omitted.  Here $C_{L\to L'}:V_L\otimes\mathbb C^2\to V_{L'}$ is the Clebsch--Gordan isometry.  Equivalently, the implementation initializes all new blocks to zero and accumulates the allowed contributions.  Normalization gives $\sum_L\operatorname{Tr}M_L=1$.

In a fixed $L$ sector the operator $X=\bS_0\cdot\bm L$ has only the two eigenvalues $x_\pm$.  Hence
\begin{equation}
e^{-\ii\lambda T X}= f_L(T)I+g_L(T)X,
\qquad
 g_L(T)=\frac{e^{-\ii\lambda T x_+}-e^{-\ii\lambda T x_-}}{x_+-x_-},
\qquad
 f_L(T)=e^{-\ii\lambda T x_+}-g_L(T)x_+.
\end{equation}
Tracing over the hub coherent state and the leaf Schur blocks gives
\begin{equation}
\cA_{\rm Schur}(T)=
\sum_L\left[
f_L(T)\operatorname{Tr}M_L
+\frac{g_L(T)}{2}\sum_{\alpha=1}^3 n_0^\alpha\operatorname{Tr}\left(M_L L^\alpha\right)
\right].
\label{eq:app-schur-amplitude}
\end{equation}
This is the exact oracle used for Fig.~\ref{fig:uniform-ll-to-quantum}.  It avoids the $2^{d+1}$ Hilbert space and is the reason the uniform-star LL-to-quantum diagnostic can be run at degrees much larger than brute-force exact diagonalization.

\subsection{Sparse static and driven oracles}

For inhomogeneous or fully driven stars no comparable Schur reduction is available generically.  The validation code therefore uses sparse Krylov propagation for small $d$.  In the static case it evaluates
\begin{equation}
\cA(T)=\langle\Psi_0|e^{-\ii HT}|\Psi_0\rangle
\end{equation}
by applying $e^{-\ii HT}$ to $|\Psi_0\rangle$ with a sparse exponential action.  In the driven case the time interval is split into midpoint steps,
\begin{equation}
U(T,0)\approx \prod_{k=1}^{n_t-1}
\exp\left(-\ii(t_{k+1}-t_k)H\left(\frac{t_{k+1}+t_k}{2}\right)\right),
\end{equation}
and each exponential action is computed by Krylov propagation.  This oracle is exponential in $d$, but it is sufficient for checking the fully driven L0/G1 formulas at $d\lesssim 10$--$12$.

\section{Numerical details}
\label{app:numerics}

All plotted errors use the branch-continuous logarithm
\begin{equation}
\log_{\rm cont}\cA(t_i)=\log|\cA(t_i)|+\ii\,\operatorname{unwrap}\arg\cA(t_i),
\end{equation}
starting from $\cA(0)=1$.  The error plotted in the figures is
\begin{equation}
E_{\rm L0}(t_i)=\left|\log_{\rm cont}\cA_{\rm exact}(t_i)-\log_{\rm cont}\cA_{\rm L0}(t_i)\right|,
\end{equation}
and similarly for L0+G1.  The logarithmic error is the natural diagnostic for the cumulant hierarchy, because L0 and G1 approximate $\log\cA$ rather than $\cA$ additively.  The validity horizons of Fig.~\ref{fig:uniform-ll-to-quantum} are defined from this error: $T_\varepsilon$ is the first time at which the branch-continuous log-error exceeds the threshold $\varepsilon=0.1$, scanned over the window $T\le 2d$ (in units of $1/J_0$, with $J_0=1$) for $d\in[12,64]$; the reported exponents are least-squares fits of $\log T_\varepsilon$ against $\log d$.

The G1 double integral is evaluated on the ordered triangle $t_i\ge t_j$.  The two branches of the time-ordered integrand are smooth on the closed triangle, and a trapezoidal rule is applied in the outer and inner variables.  This avoids differentiating across the time-ordering discontinuity.  The implementation stores two-time kernels explicitly, with memory $O(d n_t^2)$ for the star.  For larger trees one should compress these kernels by low-rank, spline, Chebyshev, or causal Volterra representations.

The driven-star scaling test of Fig.~\ref{fig:driven-scaling-ensemble} uses the nested-ensemble protocol of Sec.~\ref{sec:driven-ensemble}: per-leaf random streams keyed by (seed, leaf index), final time $T=2$, kernels on $n_t=161$ grid points (the G1 quadrature is converged to eight digits already at $n_t=81$), and the sparse midpoint oracle as the exact reference up to $d=12$.  Errors are averaged arithmetically over eight seeds; per-seed slope scatter is reported in Sec.~\ref{sec:driven-ensemble}.

\section{Temporal influence-matrix baseline}
\label{app:tmps}

This appendix gives a self-consistent construction of the baseline of Sec.~\ref{sec:tmps-baseline}.  We first recall the matrix-product representation itself, then present the four steps of the algorithm -- discretization, exact per-leaf influence MPO, compressed aggregation, and closure -- and end with the verification protocol and the parameter counting used in Fig.~\ref{fig:tmps-baseline}.

\paragraph{Matrix-product representations along the time direction.}
A matrix-product state (MPS) over $n$ sites with local dimension $q$ is a factorization of a rank-$n$ tensor into a product of site tensors,
\begin{equation}
T^{s_1s_2\cdots s_n}
=\sum_{a_1,\dots,a_{n-1}}
A^{s_1}_{1\,a_1}A^{s_2}_{a_1a_2}\cdots A^{s_n}_{a_{n-1}\,1},
\qquad s_k\in\{1,\dots,q\},
\label{eq:mps-def}
\end{equation}
where the auxiliary indices $a_k$ are called bond indices and the largest of their ranges, $\chi=\max_k\dim(a_k)$, is the \emph{bond dimension} of the representation.  The two integers -- $q$ and $\chi$ -- play different roles: the local dimension $q$ counts the physical states per site and is fixed by the model, while the bond dimension $\chi$ counts the auxiliary states carried between neighboring sites and is the adjustable resource (and accuracy knob) of the representation.  Every tensor is exactly of the form \eqref{eq:mps-def} for $\chi$ large enough.  Conversely, the optimal rank-$\chi$ truncation of any single bond is obtained by cutting the chain at that bond, viewing $T$ as a matrix between the left and right index groups, and keeping the $\chi$ largest singular values; sweeping this truncation along a suitably gauged (canonicalized) chain is the standard MPS compression algorithm \cite{verstraete_matrix_2008,orus_practical_2014,cirac_matrix_2021}.  A matrix-product operator (MPO) is the same construction for a tensor carrying a \emph{pair} of indices (incoming and outgoing) per site, i.e., an MPS of local dimension $q^2$.  In the temporal, influence-matrix use of these representations \cite{banuls_matrix_2009,strathearn_efficient_2018,lerose_influence_2021,sonner_influence_2021}, the ``sites'' are the time slices of a retained system, and the represented tensor is the discretized influence functional of an eliminated environment; the bond dimension required for a target accuracy then measures the \emph{temporal entanglement} of that influence.  In our case the retained system is the hub, each site carries the incoming and outgoing hub indices of one Trotter slice ($q^2=4$), and the environment is the set of leaves.  The matrix-product representation is logically independent of how the discrete tensor is produced; the discretization is the separate first step of the algorithm, to which we now turn.

\paragraph{Step 1: Trotter discretization.}
The propagator \eqref{eq:propagator} is factorized by a symmetric (Strang) splitting \cite{trotter_product_1959,suzuki_generalized_1976},
\begin{align}
U(T,0) & \approx\prod_{k=1}^{n_s}
U_0\!\left(\tfrac{\Delta t}{2};t_k\right)
\left[\prod_{c=1}^d G_c(\Delta t;t_k)\right]
U_0\!\left(\tfrac{\Delta t}{2};t_k\right),
\label{eq:strang}\\
G_c(\Delta t;t_k) & =\exp\!\left(-\ii\,\Delta t\left(\bS_0^{\mathsf T}K_c(t_k)\bS_c+\bb_c(t_k)\cdot\bS_c\right)\right),
\label{eq:pair-gate}
\end{align}
with $\Delta t=T/n_s$, all drives evaluated at the slice midpoints $t_k$, and $U_0$ the one-spin hub gate generated by $\bB_0(t_k)$.  This step defines the finite discrete dynamics -- the object that the tensor-network representation will treat -- and carries the usual $O(\Delta t^2)$ splitting error per slice, measured below as the Trotter offset.

\paragraph{Step 2: exact influence MPO of one leaf.}
Reshape the two-spin gate \eqref{eq:pair-gate} as a matrix in the leaf space whose entries are one-spin hub operators,
\begin{equation}
\big[W_c(k)\big]^{\,s_k's_k}_{\,l'l}
=\big\langle l'\,s_k'\,\big|\,G_c(\Delta t;t_k)\,\big|\,l\,s_k\big\rangle,
\label{eq:leaf-mpo-tensor}
\end{equation}
where $s_k,s_k'$ are the incoming/outgoing hub indices of slice $k$ (the physical indices) and $l,l'\in\{\uparrow,\downarrow\}$ label the leaf basis (the bond indices).  Chaining the tensors \eqref{eq:leaf-mpo-tensor} over $k$ contracts the leaf indices between consecutive slices -- the leaf world line -- and the boundary vectors $|\bnu_c\rangle$, $\langle\bnu_c|$ are absorbed into the first and last tensors.  The result is, with \emph{no} approximation, an MPO over the hub time slices whose bulk bond dimension equals two, the dimension of the leaf Hilbert space: a single leaf has exactly this much temporal entanglement to offer.

\paragraph{Step 3: aggregation and compression.}
The influences of the $d$ leaves act on the hub within each slice, so the combined influence is the slice-wise product of the $d$ MPOs: physical (hub) indices are matrix-multiplied, bond indices are stacked, and the bond dimension multiplies -- $2^m$ after $m$ leaves, up to $2^d$ for the full star.  This exponential growth is tamed by compression.  Whenever the running bond exceeds the preset maximum $\chi$, the MPO is viewed as an MPS of local dimension $4$ and compressed by the standard two-sweep algorithm: a first sweep of QR factorizations puts the chain into canonical gauge (so that subsequent local truncations are quasi-optimal globally), and a second sweep performs, bond by bond, the singular-value decomposition and retains the $\chi$ largest singular values, absorbing the discarded factors into the neighboring tensor.  This is precisely how the truncation from bond $2^m$ down to $\chi=1,2,3,4,\dots$ quoted in the main text is performed; the discarded singular-value weight is monitored.  The cost of one multiply-and-compress cycle is $O(n_s\chi^3)$, and $d$ cycles are needed for the star.

\paragraph{Step 4: closure.}
After all leaves are absorbed, the hub half-step gates $U_0(\Delta t/2;t_k)$ of Eq.~\eqref{eq:strang} are applied to the physical legs of each slice tensor, and the amplitude estimate is obtained by contracting the resulting MPO with the hub boundary states $\langle\bn_0|$, $|\bn_0\rangle$ at the two temporal ends, at cost $O(n_s\chi^2)$.

\paragraph{Verification.}
Two unit tests fix the implementation.  First, at $\chi=2^d$ no truncation ever occurs, and the MPO contraction must equal the amplitude obtained by applying the identical gate sequence \eqref{eq:strang} to the full $2^{d+1}$-dimensional statevector; the released code reproduces this circuit-exact reference to machine precision.  Second, the circuit-exact amplitude converges to the continuum sparse oracle as $n_s$ grows; at $n_s=100$, $d=8$, $g=1$ the Trotter offset is $6\times10^{-5}$ in the log-amplitude -- the dotted floor in panels (a) and (b) of Fig.~\ref{fig:tmps-baseline}.  Compression errors in panel (a) are measured against the circuit-exact reference (isolating them from the Trotter offset), while the coordination scan of panel (b) reports total errors of all methods against the common continuum oracle, averaged over four nested-ensemble seeds (Sec.~\ref{sec:driven-ensemble}) at $g=1$, $n_s=100$, $n_t=201$.

\paragraph{Parameter counting.}
The IM message is the list of slice tensors $\{W_k\}$: each is a $\chi_k\times\chi_{k+1}$ matrix of $2\times2$ hub blocks, i.e., $4\chi_k\chi_{k+1}$ complex numbers, for a total of $\sum_k 4\chi_k\chi_{k+1}\approx 4n_s\chi^2$ in the bulk.  The aggregate Gaussian message consists of the effective field $\bB_{\rm eff}(t)$ -- three complex components (the weak means are complex) on $n_t$ grid points, i.e., $3n_t$ numbers -- and the aggregated kernel $\sum_cK_c\kappa_cK_c^{\mathsf T}$ -- a complex $3\times3$ matrix on the $n_t\times n_t$ time grid, i.e., $9n_t^2$ numbers -- for a total of $3n_t+9n_t^2$: $9.2\times10^4$ at $n_t=101$ and $3.6\times10^5$ at $n_t=201$, compared with $6.3\times10^3$ for the IM message at $\chi=4$, $n_s=100$.  Two qualifications keep this comparison fair.  The Gaussian count is an uncompressed upper bound: smooth kernels admit the low-rank, spline, Chebyshev, or causal Volterra compressions discussed in Appendix~\ref{app:numerics}.  And, most importantly, the Gaussian count does not grow under aggregation -- leaves and subtrees are absorbed by \emph{adding} their $(\mu,\kappa)$ contributions into the same fixed-size structure -- whereas IM messages compose by MPO multiplication followed by renewed compression.

\end{document}